\documentclass[12pt]{article}
\usepackage{fullpage}
\usepackage{amsfonts, amsmath, amssymb,latexsym,epsfig}
\usepackage{amsthm}
\usepackage{thmtools}
\usepackage{tikz}
\usetikzlibrary{graphs}
\usetikzlibrary{graphs.standard}
\usepackage{relsize}
\usepackage{verbatim}
\usepackage{enumerate}
\usepackage{dsfont}
\usepackage[skip=10pt plus1pt, indent=20pt]{parskip}

\usepackage[hidelinks]{hyperref}
\usepackage[capitalise,nameinlink,noabbrev]{cleveref}
\usepackage{bbm}

\newtheorem{thm}{Theorem}[section]
\newtheorem{cor}[thm]{Corollary}
\newtheorem{lem}[thm]{Lemma}
\newtheorem{prop}[thm]{Proposition}

\theoremstyle{definition}

\newtheorem{defn}[thm]{Definition} 
 
\newtheorem{notation}[thm]{Notation} 
\newtheorem{problem}[thm]{Problem} 
\newtheorem*{remark}{Remark}

\newtheorem{manualtheoreminner}{Theorem}
\newenvironment{manualtheorem}[1]
  {\renewcommand{\themanualtheoreminner}{\ref{#1}}%
   \begin{manualtheoreminner}}
  {\end{manualtheoreminner}}

  \newtheorem{manualnotinner}{Notation}

\newcommand{\HH}{\mathcal{H}}

\newcommand{\A}{\mathcal A}

\newcommand{\Gauss}[2]{{\begin{bmatrix} #1 \\ #2 \end{bmatrix}_2}}

\newcommand{\F}{\mathbb F}
\newcommand{\Fpn}{\mathbb{F}_{2^n}}
\newcommand{\Fpnv}{\mathbb{F}_2^n}
\newcommand{\Fpmv}{\mathbb{F}_2^m}

\newcommand{\set}[1]{\left\{#1\right\}}

\DeclareMathOperator{\degalg}{\deg_{{alg}}}
\DeclareMathOperator{\RM}{RM}
\DeclareMathOperator{\tr}{tr}
\DeclareMathOperator{\spann}{span}

\crefformat{equation}{#2(#1)#3}
\crefmultiformat{equation}{(#2#1#3)}{ and~(#2#1#3)}{, (#2#1#3)}{ and~(#2#1#3)}

\title{On Reed-Muller subcodes, Grassmannian partitions and sum-free functions}

\author{Philipp Heering\footnote{Justus-Liebig-Universität Gießen, Institute for Mathematics, Arndtstraße 2,
 35392 Gießen, Germany,
 (email: \href{mailto:philipp.heering@math.uni-giessen.de}{philipp.heering@math.uni-giessen.de})}, \ 
 Christian Kaspers\footnote{
Otto von Guericke University Magdeburg, Faculty of Mathematics, Institute for Algebra and Geometry, Universitätsplatz 2, 39106 Magdeburg,
Germany (email: \href{mailto:christian.kaspers@ovgu.de}{christian.kaspers@ovgu.de})}, \
Vladislav Taranchuk\footnote{Ghent University,
Department of Mathematics: Analysis, Logic and Discrete Mathematics, Krijgslaan 281,
B-9000 Gent, Belgium,
(email:  \href{mailto:vlad.taranchuk@ugent.be}{vlad.taranchuk@ugent.be})
}}

\date{2026}

\begin{document}

\maketitle

\begin{center}
\subsection*{Abstract}
\end{center}
A function $F:\F_{2}^{n}\to \F_{2}^{m}$ is called $k$th-order sum-free if the sum of its values over any $k$-dimensional affine subspace of $\F_2^n$ is non-zero. Carlet recently introduced this notion and constructed such functions for every $2\le k\le n$. We prove that, for $2\le k\le n-2$, the existence of a (non-degenerate) $\F_{2}^{m}$-valued $k$th-order sum-free function on $\F_{2}^{n}$ is equivalent to the existence of a codimension $m$ linear subcode of the Reed-Muller code $\mathrm{RM}(n-k,n)$ with minimum distance $3\cdot 2^{k-1}$. In particular, this yields a new family of Reed-Muller subcodes that avoid all minimum weight codewords of $\mathrm{RM}(n-k,n)$, and thus have minimum distance $3/2$ times that of $\mathrm{RM}(n-k,n)$. 
We also derive new necessary conditions for the existence of $k$th-order sum-free functions and present the first nontrivial lower bound on $m$. Finally, we observe that $k$th-order sum-free functions lead to a partition of the Grassmannian of all $k$-dimensional (linear) subspaces of $\F_2^n$ into constant-dimension subspace codes. Under the assumption that functions exist that are $k$th-order sum-free for multiple values of $k$, we obtain an improved partitioning result and a stronger upper bound on the chromatic number of the Grassmann graphs.

\noindent
 \textbf{Keywords:} Reed-Muller code, subcode, $k$th-order sum-free function, almost perfect non-linear function, Grassmann graph, coloring \\
 \textbf{MSC (2020):} 
94B05, 
11T06, 
05C15

\section{Introduction}
\label{sec:introduction}
The $r$th-order Reed-Muller code $\RM(r,n)$ of length $2^n$ is the set of value vectors of all Boolean functions on $n$ variables of algebraic degree at most $r$. It is well-known that the code $\RM(r,n)$ has minimum distance $2^{n-r}$ and dimension $\sum_{i = 0}^{r}\binom{n}{i}$. Since their introduction in 1954 by Muller~\cite{muller1954} and Reed~\cite{reed1954}, these linear codes have been intensively studied and have become one of the most important and best understood families of codes \cite{MacWilliams1977TheTO}. Despite their long history, they are still highly relevant; we refer to the survey \cite{abbe2021} for an overview of recent developments.

Identifying good subcodes of Reed-Muller codes can increase the available choices of code rate and can make Reed-Muller codes more attractive in practice \cite{subcodes_recursive}. This seems to be a difficult problem which has seen considerable attention over the years \cite{carletcharpinzinoviev1998}. It has long been known that almost perfect non-linear (APN) functions produce (linear) subcodes of the extended Hamming code (which is the Reed-Muller code RM$(n-2, n)$) with minimum distance 6 and dimension $2^n - 2n - 1$~\cite{carletcharpinzinoviev1998}. Dually, one can also obtain a subcode of RM$(2, n)$ with minimum distance $3\cdot 2^{n-3}$ and dimension $\binom{n}{2} + 1$. When $n$ is even, Preparata-like codes form a class of optimal (non-linear) subcodes of the extended Hamming code, whose size is equal to the size of a linear code with dimension $2^n - 2n$ \cite{BVW, parallelisms_1973}.  The ``dual" of the Preparata code is the Kerdock code, which is a non-linear subcode of RM$(2, n)$ and has size equal to that of a linear code with dimension $\binom{n}{2} + 2$. For a deeper background on these codes we refer the reader to \cite{MacWilliams1977TheTO}. In \cite{Jamali2022LowComplexityDO, Construction_of_RM_subcodes}, Reed-Muller subcodes are obtained by removing selected rows from the generator matrix of the Reed-Muller code. Moreover, $\RM(1,n_1)\otimes \RM(1,n_2)$ is a subcode of $\RM(2,n_1+n_2)$ with minimum distance $2^{n_1+n_2-2}$, and dimension $(n_1+1)(n_2+1)$~\cite{salomon}. Recursive subproduct codes have also been studied, see \cite{subcodes_recursive}.

A function $F\colon\F_2^n\to \F_2^m$ is called $k$th-order sum-free if $\sum_{x \in A}F(x) \neq 0$ for any affine $k$-subspace $A \subseteq \F_{2}^n$. Recently, Carlet~\cite{carlet2025inverse,carlet2025generalizations} introduced these functions as a generalization of APN functions.  An APN function is a function $F \colon \Fpnv \to \Fpnv$, such that the equation $F(x+a)+F(x)=b$ has at most $2$ solutions for all non-zero $a\in \Fpnv$ and all $b\in \Fpnv$. APN functions have been extensively studied due to their optimal resistance to differential attacks in cryptography~\cite{eurocrypt-1993-2628}. It is well-known~\cite{brinkmannleander2008,hou2006} that a function $F$ is APN if and only if $\sum_{x\in A} F(x) \ne 0$ for all affine $2$-subspaces $A\subseteq \Fpnv$, so APN functions are second-order sum-free. Since their introduction, $k$th-order sum-free functions have been studied by a number of authors \cite{carlet2025tDegree,carlet2025inverse,carlet2025generalizations,ebeling2024,houzhao2025b,houzhao2025a,kaspers2026WCC}, however, to the best of the authors' knowledge, no connection to coding theory has been established yet. For more background on the properties of these functions and their connections to cryptography, we refer to~\cite{carlet2025generalizations}.

In this paper, we establish a correspondence between subcodes of $\RM(r,n)$ and $(n-r)$th-order sum-free functions which satisfy a natural non-degeneracy condition. This correspondence enables us to give explicit constructions of subcodes in $\RM(r,n)$ of codimension $n$ for each $r$ and $n$ satisfying $2 \leq r \leq n-2$. For the range $2 < r < n-2$, these codes appear to be new. The subcodes we construct have a minimum distance of $3\cdot 2^{n-r-1}$, compared with the minimum distance $2^{n-r}$ of $\RM(r,n)$. Interestingly, the codimension of these subcodes is only $n$. Our construction is inspired by the subcodes of RM$(n-2, n)$ arising from APN functions. 

\begin{thm}\label{T: Equivalence}
    Let $r,m, n$ be positive integers satisfying $2 \leq r \leq n-2$ and $m \leq n$. Any codimension $m$ (linear) subcode of the Reed-Muller code $\RM(r,n)$, which has minimum distance $3\cdot 2^{n-r-1}$, gives rise to a non-degenerate $(n-r)$th-order sum-free $(n, m)$-function. Conversely, any  non-degenerate $(n-r)$th-order sum-free $(n, m)$-function yields a codimension $m$ (linear) subcode of $\RM(r,n)$, with minimum distance $3 \cdot 2^{n-r-1}$.
\end{thm}

Besides the new subcodes of $\RM(r,n)$, our work also leads to new results on the existence of $k$th-order sum-free functions. Up to now, only one infinite family of $k$th-order sum-free functions, introduced by Carlet~\cite{carlet2025tDegree,carlet2025generalizations}, is known.  While more $k$th-order sum-free functions from $\Fpnv$ to $\Fpmv$ do exist when $m = n$, as was shown by the second author \cite{kaspers2026WCC}, it is not clear if they exist for $m <n$. Therefore, Carlet~\cite{carlet2025generalizations} asked for each $k$: What is the smallest value of $m$ such that a $k$th-order sum-free $(n,m)$-function exists? It is well known that for $m<n$, a function cannot be APN (second-order sum-free), unless $n = 2$ as was shown by Nyberg~\cite{eurocrypt-1993-2628}. Carlet~\cite{carlet2025generalizations} observed that a Boolean function cannot be $k$th-order sum-free for all $1\le k\le n-1$. For $1<m<n$ and $k \ge 3$ however, the question is completely open.  In \cref{T: m lower bound} we present the first nontrivial lower bound: $m\geq \max\{n-k+2,k+2\}$. Moreover, we show that for the only known infinite family of $k$th-order sum-free functions, we always have $m=n$.

Finally, we obtain results related to partitioning the Grassmannian of all (linear) $k$-subspaces of $\F_2^n$ into constant-dimension subspace codes, a topic which has seen much activity in recent years \cite{dhaeseleer2026chromaticnumbergrassmanngraphs, dhaeseleer2025chromaticnumbergrassmanngraphs, Etzion2015PartialKParallelisms, EtzionSilberstein2009Ferrers, heering2025lineparallelismspgn2preparatalike, KoetterKschischang2008, SilbersteinEtzion2011Lexicodes}. Our results are best stated in terms of the chromatic number of the Grassmann graphs $J_2(n, k)$ whose vertex set is the Grassmannian of all (linear) $k$-subspaces of $\F_2^n$ and two vertices are adjacent if their corresponding $k$-subspaces meet in a (linear) $(k-1)$-subspace. A valid coloring of $J_2(n, k)$ is equivalent to partitioning all $k$-spaces of $\F_2^n$ into subspace codes where each code has minimum (subspace) distance at least 4.

We observe that any $k$th-order sum-free function $F\colon \F_2^n \rightarrow \F_2^m$ can be used to give a proper coloring of $J_2(n, k)$ with $2^m - 1$ colors. This fact, together with the known constructions of $k$th-order sum-free functions, allows us to recover and extend the result of D'haeseleer and the third author \cite{dhaeseleer2025chromaticnumbergrassmanngraphs} on $\chi(J_2(n, k))$.  Of greater interest is our generalization of recent work by the first and third author \cite{heering2025lineparallelismspgn2preparatalike} where the authors use quadratic APN permutations to give optimal colorings of $J_2(n, 2)$ using $2^{n-1}-1$ colors when $n$ is even. In Section \ref{sec:coloring}, we establish Theorem \ref{T: k and k-1-order sum-free gives coloring}, which generalizes this construction. In particular, we give an explicit coloring of $J_2(n+1, k)$ using $2^m - 1$ colors which depends on the existence of a function $F:\F_2^n \rightarrow \F_2^m$ which has algebraic degree $k$ and is both $k$th-order and $(k-1)$th-order sum-free. In Section \ref{sec: multiorder sum-free functions} the existence of such functions is discussed. A small positive result provides a coloring of $J_2(6,3)$, this improves the previously best known upper bounds \cite{dhaeseleer2026chromaticnumbergrassmanngraphs, dhaeseleer2025chromaticnumbergrassmanngraphs}. 

The remainder of the paper is organized as follows. In \cref{sec:preliminaries}, we introduce all the necessary definitions and basic results about $k$th-order sum-free functions, Reed-Muller codes and the Grassmann graph needed in the course of the paper. In \cref{sec:sum-free_functions}, we present our results on $k$th-order sum-free functions. In \cref{sec:reed-muller}, we prove \cref{T: Equivalence} which is obtained via a series of lemmas and propositions. 
In \cref{sec:coloring} we demonstrate how sum-free functions can be used to give valid colorings of the Grassmann graphs.
We conclude with \cref{sec: multiorder sum-free functions}, which is a short section on some theoretical and computational results regarding multiorder sum-free functions. Finally, we present a number of interesting open questions which naturally arise as a consequence of the results obtained in this paper.

\section{Preliminaries}
\label{sec:preliminaries}

\subsection{$k$th-order sum-free functions}
\label{subsec:sum-free_functions}
Let $\F_q$ denote the finite field with $q$ elements, and denote its multiplicative group by $\F_q^*$. Throughout the paper, we identify the finite field $\Fpn$ with the vector space $\Fpnv$. We call a $k$-dimensional subspace simply a \emph{$k$-space} and a $k$-dimensional affine subspace a \emph{$k$-flat}. It is well-known that the number of $k$-spaces of $\Fpnv$ is
\[
     \Gauss{n}{k} := \frac{(2^{n}-1)(2^{n-1}-1) \cdots (2^{n-k+1}-1)} {(2^{k}-1)(2^{k-1}-1)\cdots(2-1)}
\]
and that the number of $k$-flats of $\Fpnv$ equals $2^{n-k}\Gauss{n}{k}$.
We have $\Gauss{n}{1} = 2^n-1$, the number of non-zero elements of $\F_{2^n}$. We remark that $\Gauss{n}{0} = 1$. \par

We study the behavior of $(n,m)$-functions on $k$-flats. An \emph{$(n,m)$-function} is a function~$F$ from $\Fpnv$ to $\Fpmv$. If $m=1$, we speak of an $n$-variable \emph{Boolean function} $f$. We represent an $(n,m)$-function usually by its Boolean \emph{coordinate functions} $f_1,\dots,f_m$. If $F$ is an $(n,n)$-function, we also use its univariate representation as a polynomial on $\Fpn$. For a non-zero element $v\in\Fpmv$, we call the Boolean function $f_v$ defined by $f_v(x) = v\cdot F(x)$ a \emph{component function} of $F$; here $\cdot$ denotes a scalar product on $\Fpmv$, and in univariate representation we may simply define $f_v(x)=\tr(vF(x))$, where $\tr \colon \Fpnv \to \F_2$ is the trace function $\tr(x) = \sum_{i=0}^{n-1}x^{2^i}$.

The \emph{algebraic degree} $\degalg(f)$ of a Boolean function $f$ is the degree of $f$ when given in its algebraic normal form (see \cite{carlet2021book} for more background). For an $(n,m)$-function, the algebraic degree is defined as the maximum algebraic degree of its coordinate functions. If $F$ is an $(n,n)$-function given in univariate representation, $F(x) = \sum_{i=0}^{2^n-1}a_ix^i$, then $\degalg(F)$ is the maximum $2$-weight of $i$ with $a_i \ne 0$. The \emph{derivative} $D_aF$ of an $(n,m)$-function $F$ in direction $a \in \Fpnv$ is the $(n,m)$-function $D_aF(x) = F(x+a)+F(x)$. The $\ell$th-order derivative $D_{a_1}\cdots D_{a_\ell}F$ of $F$ in direction $(a_1,\dots,a_\ell) \in (\Fpnv)^\ell$ is the composition $D_{a_1}\cdots D_{a_\ell}F(x) = (D_{a_\ell}F\circ\dots\circ D_{a_1}F)(x)$.

Moreover, we introduce three important equivalence notions of $(n,m)$-functions. Two $(n,m)$-functions $F,G$ with graphs $G_F := \set{(x,F(x)) : x\in\Fpnv}$ and $G_G$, respectively, are called \emph{CCZ-equivalent} if there exists an affine $(n+m,n+m)$-permutation $C$ such that $G_F = C\circ G_G$. They are said to be \emph{degree-$r$ equivalent} if there exists an affine $(n,n)$-permutation $A_1$, an affine $(m,m)$-permutation~$A_2$, and an $(n,m)$-function $A_3$ of algebraic degree $r$ such that $F = A_2 \circ G \circ A_1 + A_3$.  If $F,G$ are degree-$1$ equivalent, then they are called \emph{EA-equivalent}. Clearly, EA-equivalence implies CCZ-equivalence, but in general the converse is not true. 
The notion of degree-$r$ equivalence was only recently introduced by the second author~\cite{kaspers2026WCC}, but we remark that this equivalence has been used before to classify Boolean functions~\cite{hou1996,langevinleander2008}.

We now introduce $k$th-order sum-free functions beginning with the most popular case: APN functions. An \emph{almost perfect non-linear (APN) function} is an $(n,n)$-function $F$ with the property that the equation
\[
    F(x+a)+F(x) = b
\]
has at most two solutions for all non-zero $a\in\Fpnv$ and all $b\in\Fpnv$. In \cite{brinkmannleander2008,hou2006}, it was shown that the APN property of $F$ is equivalent to $\sum_{x\in A}F(x)\ne0$ for all $2$-flats $A\subseteq\Fpnv$. This motivates the following definition recently introduced in~\cite{carlet2025inverse,carlet2025generalizations}.
\begin{defn}
\label{def:sum-free}
    Let $k,m,n$ be nonnegative integers with $k \le n$ and $n\ne 0$. An $(n,m)$-function $F$ is called \emph{$k$th-order sum-free} if for all $k$-flats $A\subseteq\Fpnv$, we have
    \[
        \sum_{x\in A}F(x) \ne 0.
    \]
\end{defn}
We remark that, unlike Carlet~\cite{carlet2025generalizations}, we also include $k=0$ and $k=1$ in the definition of $k$th-order sum-free functions. Clearly, any function with no roots is zero-order sum-free, permutations are first-order sum-free and APN functions are second-order sum-free. The APN property of a function is well-known to be invariant under CCZ-equivalence. However, to date, $k$th-order sum-freedom for $k\ge 3$ is only known to be invariant under degree-$(k-1)$ equivalence~\cite{kaspers2026WCC}. Note that there is a close connection between the sum-freedom of $F$ and its higher-order derivatives: If $A = \text{span}(a_1,\dots,a_k)+x$ is a $k$-flat of $\Fpnv$, then $D_{a_1}\cdots D_{a_k}F(x) = \sum_{y\in A}F(y)$. So $F$ is $k$th-order sum-free if and only if all its $k$th-order derivatives $D_{a_1}\cdots D_{a_k}$, with $a_1,\dots,a_k$ linearly independent, never map to zero. In \cref{Notation: Witness}, we introduce a helpful notation to speak about sums of images over flats.

\begin{defn}
\label{Notation: Witness}
    Let $F$ be an $(n,m)$-function. Denote the set of all flats of $\Fpnv$ by $\A$. We define the \emph{witness function $\omega_F \colon \A \to \Fpmv$ of $F$} by 
    \begin{align*}
        \omega_F(A)=\sum_{x\in A} F(x)
    \end{align*}
    and call $\omega_F(A)$ the \emph{witness of $A$ with respect to $F$}.
    If the function $F$ is clear from the context, we simply write $\omega(A)$.
\end{defn}

Using the language from \cref{Notation: Witness}, a function $F$ is $k$th-order sum-free if $\omega_F(A)\ne 0$ for all $k$-flats $A \subseteq \Fpnv$. Carlet~\cite{carlet2025generalizations} observed the following two properties of $k$th-order sum-free functions. We remark that his results for $k\ge 2$ naturally extend to the case $k=1$.

\begin{prop}
\label{prop:k-1_flats}
    Let $F$ be an $(n,m)$-function, and let $1\le k\le n$. Then $F$ is $k$th-order sum-free if and only if for all $(k-1)$-spaces $U \subseteq \F_2^n$ and $a, b \in \F_{2}^n$ we have $\omega_F(U+a) \neq \omega_F(U+b)$ whenever $U+a \neq U+b$.
\end{prop}
As we will see in \cref{sec:coloring}, this property allows us to partition the Grassmannian of all $k$-spaces in $\F_2^n$ where two spaces in the same part intersect in dimension at most $k-2$.

\begin{prop}
\label{prop:degree<k}
    Let $F$ be an $(n,m)$-function, and $1\le k\le n$. We have $\omega_F(A)=0$ for all $k$-flats $A\subseteq\Fpnv$ if and only if $F$ has algebraic degree at most $k-1$. In particular, if $F$ is $k$th-order sum-free, then $\degalg(F)\ge k$. 
\end{prop}

To obtain a good understanding of the Reed-Muller subcodes which arise via $k$th-order sum-free functions, as we will show in \cref{sec:reed-muller}, we require a notion of degeneracy of $(n, m)$-functions which we give next.

\begin{defn}
    Let $r\le n$. An $(n, m)$-function $F$ is called \textit{degree-$r$ degenerate} if it has a component function whose algebraic degree is less than $r$. A $k$th-order sum-free function $F$ is called \textit{degenerate} if it is degree-$k$ degenerate. Otherwise, we call $F$ \textit{non-degenerate}.
\end{defn}

\begin{remark}
    The above definition enables us to exclude $k$th-order sum-free $(n, m)$-functions which are degree-$(k-1)$ equivalent to an $(n, \ell)$-function with $\ell < m$, and hence are actually  $k$th-order sum-free $(n,\ell)$-functions in hiding. This will be helpful in \cref{sec:sum-free_functions,sec:reed-muller}.
\end{remark}

So far, when $k\ge 2$, only few $k$th-order sum-free functions are known \cite{carlet2025tDegree,carlet2025generalizations,kaspers2026WCC}. Note that any permutation is first-order sum-free, and any function which has no roots is zero-order sum-free. In both cases, there are infinitely many $(n, n)$-functions. Carlet~\cite{carlet2025tDegree,carlet2025generalizations} introduced the following infinite family of power functions which produces $k$th-order sum-free $(n,n)$-functions for all $n$ and $k$. 

\begin{prop}
\label{prop:Carlet_kth-order}
	Let $k\le n$. Define the $(n,n)$-function $F_{k,j}$ by
	\[
		F_{k,j}(x) = x^\frac{2^{jk}-1}{2^j-1} = x^{1+2^j+\dots+2^{j(k-1)}}.
	\]
	If $\gcd(j,n) = 1$, then $F_{k,j}$ is $k$th-order sum-free. In particular, the function $F_{k,1}(x) = x^{2^k-1}$ is $k$th-order sum-free.
\end{prop}

Until now, this is the only known infinite family of $k$th-order sum-free functions. In \cref{T: each Carlet function is nn}, we demonstrate that these functions are non-degenerate. Constructing genuinely new sum-free functions appears to be difficult. Currently, one of the most useful tools we have for constructing new sum-free functions from old ones is a duality between certain $k$th-order and $(n-k)$th-order sum-free functions \cite[Theorem~5.12]{kaspers2026WCC}.

\begin{prop}
\label{P: k gives n-k}
    Any $k$th-order sum-free $(n,m)$-function of algebraic degree $k$ implies the existence of an $(n-k)$th-order sum-free $(n,m)$-function of algebraic degree $n-k$.
\end{prop}

 Finally, we introduce the new term multiorder sum-free functions. As we will see, such functions will at times allow us to produce better partitions of the Grassmannian of $k$-spaces over $\F_2^{n}$.
\begin{defn}
\label{def:multiorder-sum-free}
	Let $F$ be an $(n,m)$-function, and denote
	\[
		K_F = \set{k \in \set{1, 2,\dots,n} : \text{$F$ is $k$th-order sum-free}}.
	\]
	We say that $F$ is \emph{$K_F$-order sum-free}. If $|K_F| \ge 2$, we say that $F$ is \emph{multiorder sum-free}.
\end{defn}
	
One example of a multiorder sum-free function is the inverse function $F_{inv}(x) = x^{2^n-2}$ in odd dimension $n$. Then $K_{F_{inv}} = \set{1, 2,n-2,n-1}$, see \cite{carlet2025inverse}. In \cref{th:multiorder_carlet}, we show that some of the functions from \cref{prop:Carlet_kth-order} are also multiorder sum-free.

\subsection{Reed-Muller codes}
\label{subsec:reed-muller}

\begin{defn} 
The binary Reed-Muller code $\mathrm{RM}(r,n)$ of order $r$ and length $2^n$ is the set of evaluation vectors of all $n$-variable Boolean functions of algebraic degree at most $r$:
\[
\mathrm{RM}(r,n)
=
\left\{
\bigl(f(a)\bigr)_{a \in \mathbb{F}_2^n}
\;:\;
f \in \mathbb{F}_2[x_1,\dots,x_n],\ \deg(f)\le r
\right\}.
\]
Here the coordinates of the codewords are indexed by the points of $\mathbb{F}_2^n$.
\end{defn}
The Reed-Muller code $\RM(r,n)$ is a linear code. It has dimension $\sum_{i=0}^r \binom{n}{i}$ and minimum weight $2^{n-r}$, and we have $\RM(r,n)^\perp = \RM(n-r-1,n)$. The weight of a codeword in $\RM(r,n)$ whose weight is less than $2\cdot 2^{n-r}$ is  $2^{n-r+1} - 2^{n - r + 1 - d}$ for some positive integer~$d$ as was shown by Berlekamp and Sloane \cite{1970Berlekamp}. In \cite{kasami1970}, Kasami and Tokura characterized all such codewords. We state their result in \cref{T: min-weights}.

\begin{thm}\label{T: min-weights}
    Let $f \in \RM(r, n)$ be a codeword (Boolean function) of weight strictly less than $2^{n-r+1}$. Then $f$ is equivalent to one of the following codewords $f'$ by an invertible affine transformation:
    \begin{enumerate}
        \item  $f' = x_1x_2\cdots x_{r-d}(x_{r-d+1}\cdots x_r + x_{r+1}\cdots x_{r+d})$,
        where $n \geq r + d$ and $r \geq d \geq 3$, or
        \item $f' = x_1x_2\cdots x_{r-2}(x_{r-1}x_{r} + x_{r+1}x_{r+2} + \cdots + x_{r+2d -3}x_{r+2d - 2})$, where $n - r + 2 \geq 2d \geq 2$.
    \end{enumerate}
    In either case, $f$ has weight $2^{n-r+1}-2^{n-r+1-d}$.
\end{thm}

Our equivalence result will rely on the following lemma which is a corollary of \cref{T: min-weights}. In particular, we specialize \cref{T: min-weights} to the case that $d = 1, 2$. We add a short proof for completeness.

\begin{lem}\label{L: min-weights}
    Let $r, n$ be positive integers satisfying $2 \leq r \leq n-2$. The minimum weight codewords of $\RM(r,n)$ have weight $2^{n-r}$ and are precisely the incidence vectors of $(n-r)$-flats in $\F_2^n$. The second minimum weight codewords of $\RM(r,n)$ have weight $3\cdot 2^{n-r-1}$ and are the incidence vectors of the symmetric difference of two $(n-r)$-flats in $\F_2^n$ intersecting in an $(n-r-2)$-flat.
\end{lem}
\begin{proof}
    In this proof, we identify codewords of $\RM(r,n)$ with their corresponding Boolean functions. Moreover, denote by $e_i$ the $i$th standard basis vector of $\Fpnv$. Let $f$ be a minimum weight codeword of $\RM(r,n)$, so $d=1$ in \cref{T: min-weights}. It follows from case~2 of \cref{T: min-weights} that $f$ is affine equivalent to the Boolean function $f' = x_1\cdots x_r$. This is precisely the incidence function of the $(n-r)$-flat $A_{f'} = \spann(e_{r+1},\dots,e_n)+\sum_{i=1}^r e_i$. Since any affine transformation of $A_{f'}$ gives another $(n-r)$-flat of $\Fpnv$, $f$ is also the incidence function of an $(n-r)$-flat.\par

    Now let $f$ be a second minimum weight codeword of $\RM(r,n)$, so $d=2$. Then case 2~of \cref{T: min-weights} implies that $f$ is affine equivalent to the Boolean function $f'=x_1\cdots x_{r} + x_1\cdots x_{r-2}x_{r+1}x_{r+2}$. Write $f_1' = x_1\cdots x_r$ and $f_2' = x_1\cdots x_{r-2}x_{r+1}x_{r+2}$. Then $f_1'$ and $f_2'$ are the incidence functions of the $(n-r)$-flats $A_{f_1'} = \spann(e_{r+1},\dots,e_n) + \sum_{i=1}^r e_i$ and $A_{f_2'} = \spann(e_{r-1},e_{r},e_{r+3},\dots,e_n) + (\sum_{i=1}^{r-2} e_i + e_{r+1} + e_{r+2})$, respectively. Clearly $A_{f_1'}$ and $A_{f_2'}$ intersect in the $(n-r-2)$-flat $\spann(e_{r+3},\dots,e_n) + \sum_{i=1}^{r+2} e_i$, and since $f' = f_1' + f_2'$, the function $f'$ is precisely the incidence function of the symmetric difference of $A_{f_1'}$ and $A_{f_2'}$. The result now follows by the same arguments as above.
\end{proof}
 
It is immediate that any subcode of $\RM(r,n)$ which avoids all minimum weight codewords has minimum weight at least $3\cdot 2^{n-r-1}$.

\subsection{Grassmann graphs}
\label{subsec:graphs}
For integers $0 \le t \leq  k < n$, we consider the \emph{generalized Grassmann graph $J_q(n,k,t)$}, whose vertices are the $k$-spaces of $\F_q^n$. Two distinct vertices $U,W$ are adjacent if $\dim(U \cap W) \ge t$.
We call $J_q(n,k,k-1)$ the \emph{Grassmann graph}, which we denote by $J_q(n,k)$.
First, we give two standard definitions that can be found for example in \cite{algebraic_graph_theory}.

\begin{defn}
    If $\mathcal{C}$ is a set of vertices of a graph, such that no two vertices in $\mathcal{C}$ are adjacent, then $\mathcal{C}$ is called a \emph{coclique} or an \emph{independent set}.
\end{defn}

\begin{defn}
    A partition of the vertex set of a graph $G$ into cocliques is called a \emph{(proper) coloring} of the graph where the different color classes are the different cocliques. The smallest number of colors needed for such a coloring is called the \emph{chromatic number} of $G$, we denote it by $\chi(G)$.
\end{defn}

The generalized Grassmann graph has a natural connection to constant-dimension codes. A \emph{constant $k$-dimensional subspace code} is a set of $k$-spaces of $\F_q^n$ equipped with \textit{the} subspace distance
$$d_S(U,W)=2\bigl(k-\dim(U\cap W)\bigr).$$
If a code has minimum subspace distance at least $2(k-t+1)$, then any two distinct codewords satisfy $\dim(U\cap W)\le t-1$, so no two codewords are adjacent in $J_q(n,k,t)$. In other words, a constant-dimension code with minimum distance at least $2(k-t+1)$ is precisely a coclique in $J_q(n,k,t)$. It follows that a coloring of $J_q(n,k,t)$ is a partition of the entire set of $k$-spaces into constant-dimension codes with minimum distance at least $2(k-t+1)$. The chromatic number $\chi(J_q(n,k,t))$ therefore measures how many such codes are needed to decompose all $k$-spaces of $\F_q^n$. The study of the chromatic number of the generalized Grassmann graphs was only recently initiated in \cite{dhaeseleer2026chromaticnumbergrassmanngraphs}.

The Grassmann graphs exhibit an interesting duality which allows one to study the same question from two different perspectives. In particular, it is known that $J_q(n,m,t)\simeq J_q(n,n-m,n-2m+t)$ (see \cite{dhaeseleer2026chromaticnumbergrassmanngraphs} for a proof). The results in \cite{dhaeseleer2026chromaticnumbergrassmanngraphs} establish a lower bound for $\chi(J_q(n, k, t))$ for $1 \leq t < k \leq n$. As some of our results (\cref{prop:m ge n-k+1}, \cref{P: Code bound}) rely on this lower bound, we give the following proposition which restates this result for $q = 2$ and extends it to the edge cases where $t = 0$ or $t = k$.

\begin{prop}
\label{P: general bound chromatic number}
Let $1\leq t \leq k < n$ be integers. Then
\[
 \max \left\{ \Gauss{n-t}{ k-t },  \Gauss{2k-t}{ k-t } \right\} \leq \chi(J_2(n, k, t)) .
\]
For $t=0$, we have $\chi(J_2(n, k, t))=\Gauss{n}{k}$.
\end{prop}

\begin{proof}
    When $1 \leq t < k < n$, this result is precisely given by \cite[Proposition 3.1]{dhaeseleer2026chromaticnumbergrassmanngraphs}. The only remaining cases are:
    \begin{enumerate}
        \item $t= 0$: In this case, the graph $J_2(n, k, 0)$ is the complete graph on $\Gauss{n}{k}$ vertices. Its chromatic number is the same as the number of vertices of the graph. 
        \item $t = k$: In this case, the graph $J_2(n, k, k)$ is the empty graph on $\Gauss{n}{k}$ vertices and has chromatic number 1. Again, the lower bound holds and is tight. \qedhere
    \end{enumerate} 
\end{proof}
The results of this paper have direct implications for the chromatic number of $J_2(n, k)$. The best-known general bounds for $q = 2$ follow from the work in \cite{dhaeseleer2026chromaticnumbergrassmanngraphs,dhaeseleer2025chromaticnumbergrassmanngraphs}.

\begin{prop}\label{P: bounds_chromatic}
    Let $k < n$ be positive integers. Then
    $$
\max \left\{\Gauss{k+1}{1}, \Gauss{n-k+1}{1} \right \} \leq \chi(J_2(n, k)) \leq \Gauss{n}{1}.
$$
\end{prop}

The upper bound above specifically comes from \cite{dhaeseleer2025chromaticnumbergrassmanngraphs}, and we highlight that the method the authors used to prove this coincides with Carlet's work on $k$th-order sum-free functions~\cite{carlet2025generalizations}. In particular, the function $F_{k, 1}$ from \cref{prop:Carlet_kth-order} is shown to be $k$th-order sum-free by relating the sum of $F_{k,1}$ over a $k$-space to a special determinant. This same determinant (and its generalization over $\F_q$) is exactly what was used in \cite{dhaeseleer2025chromaticnumbergrassmanngraphs} to attain the upper bound for $\chi(J_q(n, k))$. We remark that the upper bound of \cref{P: bounds_chromatic} will become a special case of \cref{T: kth-order coloring}.

\begin{lem}\label{L: kth-order coloring}
    Let $k < n$ be positive integers and $F$ be a $k$th-order sum-free $(n, m)$-function. Then for any two distinct $k$-flats $A_1, A_2 \subset \F_2^n$ that intersect in a $(k-1)$-flat, we have $\omega_F(A_1)\neq \omega_F(A_2)$. 
\end{lem}

\begin{proof}
Let $F$ be a $k$th-order sum-free $(n, m)$-function.
Since $A_1$ and $A_2$ intersect in a $(k-1)$-flat, there exists a subspace $U$ and vectors $a, b, c \in \F_2^n$ such that $A_1 = (U+a)\cup(U+b)$ and $A_2 = (U+a)\cup (U+c)$. By \cref{prop:k-1_flats}, we have that $\omega(U+b) \neq \omega(U+c)$. Hence,
\begin{equation*}
    \omega(A_1) = \omega(U+a) + \omega(U + b)\neq \omega(U+a) + \omega(U+c) = \omega(A_2).
 \qedhere
\end{equation*}
\end{proof}

\begin{prop}\label{T: kth-order coloring}
    Let $k < n$ be positive integers. Assume there exists a $k$th-order sum-free $(n,m)$-function. Then there exists a coloring of $J_2(n,k)$ with $\Gauss{m}{1}$ colors.
\end{prop}

\begin{proof}
    Let the colors be the non-zero elements of $\F_2^m$. Assign to each $k$-space $U$ the color which is given by its witness $\omega(U)$ under $F$. \cref{L: kth-order coloring} implies that if $U, V \subset \F_2^n$ are two $k$-spaces which intersect in a $(k-1)$-space, then they have distinct witnesses. Hence, we have a valid coloring of $J_{2}(n, k)$ using $\Gauss{m}{1}$ colors.
\end{proof}

The only case where the chromatic number of $J_2(n, k)$ is known precisely is for $k = 2$. In particular, the results of \cite{Meszka2013, parallelisms_1973}  imply that 
\[
\chi(J_{2}(n, 2)) = 
\begin{cases}
    \Gauss{n-1}{1} &  \text{if $n$ is even,}\vspace{.5em}\\
    \Gauss{n-1}{1} + 3 &  \text{if $n$ is odd.}
\end{cases}
\]
When $n$ is even, a number of inequivalent colorings are known, see \cite{heering2025lineparallelismspgn2preparatalike} for a summary of known colorings. The largest class of optimal colorings when $n$ is even arises from the existence of special functions called \textit{crooked functions}~\cite{heering2025lineparallelismspgn2preparatalike}.
Though we will not go into a detailed discussion of crooked functions, we highlight that a sufficient condition for crookedness is that the function is a quadratic APN permutation. In the terms of this paper, the main result of \cite{heering2025lineparallelismspgn2preparatalike} states that $(n-1, n-1)$ functions which are quadratic and $\set{1,2}$-order sum-free can produce optimal colorings of $J_2(n, 2)$ using $\Gauss{n-1}{1}$ colors.

The key result we obtain in \cref{sec:coloring} is a generalization of the coloring result in \cite{heering2025lineparallelismspgn2preparatalike}.  In particular, \cref{T: k and k-1-order sum-free gives coloring} gives an explicit coloring of $J_2(n, k)$ using $\Gauss{m}{1}$ colors when a $\{k-1, k\}$-order sum-free $(n-1, m)$-function of algebraic degree $k$ exists.

\section{On the existence of $k$th-order sum-free functions}
\label{sec:sum-free_functions}

In this section, we present a lower bound on $m$ for which an $(n,m)$-function can be $k$th-order sum-free and thereby give a partial answer to a question by Carlet~\cite{carlet2025generalizations}. Note that $(n, m)$-functions which are $n$th-order sum-free (so $k=n$) exist for any $m \geq 1$. Hence, we focus on the case that $k < n$. We begin by establishing a first lower bound on $m$ in \cref{prop:m ge n-k+1}, whose proof is straightforward. We then improve upon this with the help of \cref{L: n to n-j} and \cref{L: kth-order-kth-degree} to obtain a strong bound in \cref{T: m lower bound}. The proof of this theorem relies on both the non-existence of APN $(n, n-1)$-functions and on \cref{P: k gives n-k}. We conclude this section by showing that the $k$th-order sum-free functions from \cref{prop:Carlet_kth-order} are non-degenerate.

\begin{prop}
\label{prop:m ge n-k+1}
    Let $k<n$ be positive integers and let $F$ be an $(n,m)$-function. If $F$ is $k$th-order sum-free, then $m \ge \max\{n-k+1, k+1\}$. 
\end{prop}
\begin{proof}
    Let $F$ be a $k$th-order sum-free $(n, m)$-function. By \cref{T: kth-order coloring} we can obtain a coloring of all $k$-spaces of $\F_{2}^n$ using $2^m-1$ colors such that any two $k$-spaces which meet in a $(k-1)$-space receive distinct colors. Hence, we have a valid coloring of $J_2(n, k)$. By \cref{P: bounds_chromatic}, we have that $\max\{ n-k+1, k+1\}\leq m$. 
\end{proof}

\begin{lem}\label{L: n to n-j}
Let $j$ be a nonnegative integer and $k, m, n$ be positive integers satisfying $k+j \le n$. If a $(k+j)$th-order sum-free $(n, m)$-function exists, then a $k$th-order sum-free $(n-j, m)$-function exists.
\end{lem}

\begin{proof}
    For $j=0$, the lemma is trivial. Let $j > 0$, and suppose that $F$ is a $(k+j)$th-order sum-free $(n, m)$-function. Let $v_1, \dots, v_j \in \F_{2}^n$ be linearly independent and consider the $j$th-order derivative $D_{v_1}\cdots D_{v_j}F$ of $F$ in direction $(v_1, \dots, v_j)$. Let $V = \text{span}(v_1, \dots, v_j)$ and $W$ be any $(n-j)$-space of $\Fpnv$ satisfying $V \cap W = \{ 0 \}$. We show that the restriction of $D_{v_1}\cdots D_{v_j}F$ to $W$ is a $k$th-order sum-free $(n-j, m)$-function. By construction, the restricted function is an $(n-j, m)$-function. If $A \subseteq W$ is a $k$-flat, then $V + A := \set{v+a : v \in V, a \in A}$ is a $(k +j)$-flat in $\Fpnv$. Therefore, 
    $$
    \sum_{a \in A}D_{v_1}\cdots D_{v_j}F(a) = \sum_{a \in A}\sum_{v \in V}F(v + a) = \sum_{x \in V+A} F(x) \neq 0
    $$ 
     since $F$ is $(k+j)$th-order sum-free.
\end{proof}

\begin{lem}\label{L: kth-order-kth-degree}
    Let $1\le k<n$, and let $F$ be an $(n,k+1)$-function. If $F$ is $k$th-order sum-free, then $\degalg(F) = k$.
\end{lem}
\begin{proof}
    To prove this result, we show that $F$ sums to 0 over any $(k+1)$-flat of $\Fpnv$. \cref{prop:degree<k} then implies that $F$ has algebraic degree $k$. Let $A\subseteq \Fpnv$ be a $(k+1)$-flat. Write $A = U+a$ for some $(k+1)$-space $U\subseteq\Fpnv$ and $a \in \F_{2}^n$. Denote by $S$ the set of all $k$-spaces of $U$. For two distinct $k$-spaces $V, W \in S$, the intersection $(V+a) \cap (W+a)$ is a $(k-1)$-flat. It follows from \cref{L: kth-order coloring} that $\omega(V+a) \neq \omega(W+a)$. So for fixed $a$, the witnesses $\omega(V+a)$ with $V \in S$ are distinct, and they are non-zero because $F$ is $k$th-order sum-free. The number of $k$-spaces $V$ in $S$ is $\Gauss{k+1}{k}=\Gauss{k+1}{1}=2^{k+1}-1$, so $|S|=2^{k+1}-1$. Therefore, the witnesses are exactly the non-zero elements of $\F_2^{k+1}$. Hence,
    \[
    \sum_{V \in S}\omega(V+a) = \sum_{y \in \F_{2}^{k+1}\setminus\set{0}}y = 0.
    \]
  Now, let $x\in A$. If $x=a$, any $k$-flat $V+a$ with $V\in S$ contains $x$. Hence, the number of $k$-flats $V+a$ with $V\in S$ that contain~$x$ is $|S|=2^{k+1}-1$ in this case. If $x\neq a$, a $k$-flat $V+a$ contains $x$ only if $V$ contains $x+a$. The number of $k$-spaces in the $(k+1)$-space $U$ that contain a fixed vector is $\Gauss{(k+1)-1}{k-1}=\Gauss{k}{1}=2^k-1$. Hence, the number of $k$-flats $V+a$ with $V\in S$ that contain~$x$ is $2^k - 1$ in this case.
    Therefore,
    \[
    0 = \sum_{V \in S} \omega(V + a) = (2^{k+1} - 1)F(a) + \sum_{y \in A\setminus \{ a \}} (2^k - 1)F(y) = \sum_{y \in A} F(y), 
    \]
    which concludes the proof.
\end{proof}

In particular, \cref{L: kth-order-kth-degree} implies the following. 
\begin{cor}
    Any $(n-1)$th-order sum-free $(n,n)$-function has algebraic degree $n-1$.
\end{cor}

We now prove our first theorem.
 
\begin{thm} \label{T: m lower bound}
    Let $k, m, n$ be positive integers with $2 \leq k \leq n-2$. If $F$ is a $k$th-order sum-free $(n, m)$-function, then $m \ge \max\{n-k+2, k+2\}$. 
\end{thm}

\begin{proof} 
    We show that equality in the bound of \cref{prop:m ge n-k+1} cannot hold. First, suppose that $n -k \geq k$. Since $k \geq 2$, \cref{L: n to n-j} implies that if a $k$th-order sum-free $(n, n-k+1)$-function exists, then an APN $(n-k+2, n-k+1)$-function exists. But such functions do not exist unless $n - k + 1= 1$ or, equivalently, $n = k$. Hence if $k\le n-2$, there are no $k$th-order sum-free $(n, n-k+1)$-functions. 

    Second, let $k > n- k$, and suppose that $F$ is a $k$th-order sum-free $(n, k+1)$-function. According to \cref{L: kth-order-kth-degree}, $F$ has algebraic degree $k$. By \cref{P: k gives n-k}, this implies that there exists an $(n-k)$th-order sum-free $(n, k+1)$-function. Set $j = n - k$ and note that since $k \leq n-2$, we have $j \geq 2$. We now repeat the first argument using $j$ instead of $k$ and obtain that $j$th-order sum-free $(n, n-j+1)$-functions do not exist. Consequently, there are also no $k$th-order sum-free $(n, k+1)$-functions.
\end{proof}

We now show that the known examples of $k$th-order sum-free power functions $F_{k, j}$ from \cref{prop:Carlet_kth-order} are non-degenerate if $k<n$. Recall that 
    \[
    F_{k, j} = x^{\frac{2^{jk} - 1}{2^j-1}} = \left(x^{\frac{1}{2^j - 1}}\right)^{2^{jk}-1}.
    \]
    with $\gcd(j, n) = 1$. If $k=n$, then $F_{n,1}$ is degenerate as $F_{n,1}(x) = x^{2^n-1}$, and $F_{n,1}$ has image~$\F_2$.

\begin{thm}
\label{T: each Carlet function is nn}
  Let $k < n$ be positive integers. Each function $F_{k, j}$ in the family of $k$th-order sum-free functions from \cref{prop:Carlet_kth-order} is non-degenerate. 
\end{thm}

\begin{proof}
    We first prove that the image set of $F_{k, j}$ is not contained in any proper subspace of $\F_{2^n}$. With this fact, it will quickly follow that $F_{k, j}$ is non-degenerate.

    Since $\gcd(j, n) = 1$, then $\gcd(2^j -1, 2^n - 1) = 1$. Hence, 
    $x^{\frac{1}{2^{j}-1}}$ is a permutation of $\F_{2^n}$ and so the image set of $F_{k, j}$ is the same as the image set of $F_{jk, 1} = x^{2^{jk}-1}$. The remainder of the proof only distinguishes cases when $\gcd(2^{jk}-1, 2^n -1)$ is distinct. Hence, without loss of generality, we assume that the function has the form $F_{k, 1}$.
    
     Let $I_k = \{ x^{2^{k}-1}: x \in \F_{2^n}\}$. We will show that $I_k$ is not contained in any proper subspace of $\F_{2^n}$. To prove this, it is sufficient to demonstrate that there is some element $\alpha \in I_{k}$ which is not contained in any proper subfield of $\F_{2^n}$. In particular, if $\alpha \in \F_{2^n}$ is not contained in any proper subfield, then $\alpha$ is a root of an irreducible polynomial of degree $n$ over $\F_2$. Hence, $\{ 1, \alpha, \alpha^2, \dots, \alpha^{n-1} \} $ forms a basis for $\F_{2^n}$ over $\F_2$. Now, if such an $\alpha$ is contained in $I_k$, then $\{1, \alpha, \alpha^2, \dots, \alpha^{n-1} \} \subseteq I_k$, because $I_k\setminus \{ 0 
      \}$ forms a cyclic subgroup of $\F_{2^n}^*$. This implies that the image of $F_{k, 1}$ is not contained in any proper subspace of $\F_{2^n}$.
      
    The number of elements in $\F_{2^n}$, which are contained in a proper subfield, is at most $2|\F'| -1 $, where $\F'$ is the largest proper subfield of $\Fpn$. Indeed, we have one distinct proper subfield of $\F_{2^n}$ for each proper divisor of $n$. If $d$ is the largest proper divisor, then the total number of elements contained in any proper subfield is certainly less than
    \[
        2^d + 2^{d-1} + \cdots +2 + 1 = 2^{d+1}-1.
    \]
    
    Our goal is to show that $2^{d+1}-1\leq |I_k|$ when $k \neq n/2$. This will imply that $I_k$ does contain an element not contained in any proper subfield of $\Fpn$. Observe that $|I_k| = (2^n-1)/(2^{\gcd(n,k)}-1) + 1$ which implies $2^{n/2}+2\leq |I_k|$ for each $k$.
    If $n$ is odd, then $d\le n/3$, and we always have $2^{d + 1}-1\leq 2^{n/2} +2$. 
    So now, let $n$ be even and therefore $d = n/2$. 
    Since $k \neq n/2$, then $\gcd(n, k) \leq n/3$. Then, $2^{2n/3}+2^{n/3}+2\le  (2^n-1)/(2^{\gcd(n,k)}-1) + 1=|I_k|$. As $2^{n/2+1}-1\leq 2^{2n/3}+2^{n/3}+2$, $I_k$ is not contained in any proper subspace of $\F_2^n$ in each of these cases. 
    
    Thus, we are left only with the case that $n$ is even and $k = n/2$. For $n=2$ and $k=1$, the function $F_{1,1}$ is a permutation and $I_1 = \F_4$. So suppose $n \ge 4$. Consider the image~$I_{n/2}$ of $x^{2^{n/2}-1}$ which is precisely the set of $2^{n/2} + 2$ roots of $x(x^{2^{n/2} + 1} - 1)$. We claim that the only elements of $I_{n/2}$ which are contained in $\F_{2^{n/2}}$ are the elements $0$ and $1$. The non-zero elements in $I_{n/2}\cap \F_{2^{n/2}}$ are of the form $y = x^{2^{n/2} - 1}$ and must satisfy
    \begin{equation}\label{EqInt}
    1 = y^{2^{n/2} - 1} = \left( x^{2^{n/2}-1}\right)^{2^{n/2}-1} = x^{2(1 - 2^{n/2})}.      
    \end{equation} 
    Therefore, if $y \in I_{n/2}\cap\F_{2^{n/2}}$, then $y = x^{2^{n/2} - 1}$ for some $x \in \F_{2^n}$ and (\ref{EqInt}) implies that $x^{-2} \in \F_{2^{n/2}}$. Now, $x^{-2} \in \F_{2^{n/2}}$ if and only if $x \in \F_{2^{n/2}}$, and hence $y = 1$. Therefore, $I_{n/2}\cap \F_{2^{n/2}} = \{0, 1\}$.
     But now, $I_{n/2}\setminus\{ 0, 1\}$ still has $2^{n/2}$ elements which must all be contained in a subfield, but are not contained in the largest subfield $\F_{2^{n/2}}$. The union of the remaining subfields has size at most $2|\F''| - 1$, where $\F''$ is the next largest subfield of $\F_{2^n}$. It follows that $|\F''| \leq 2^{n/2 - 1}$ and since $\F_2 \subseteq \F''$, then $I_{n/2}$, which has cardinality $2^{n/2} + 2$,
     contains some elements which are not contained in any proper subfield. Thus we have shown that the image of $F_{k,j}$ is not contained in any proper subspace of $\F_{2^n}$ for any valid choice of $k$ and $j$.

     Now, any component function of $F_{k, j}$ is given by $\tr(aF_{k, j})$ for some $a \in \F_{2^n}^*$. Because $F_{k, j}$ is a monomial, then either $\tr(aF_{k, j}) = 0$, or the algebraic degree of $F_{k, j}$ is the same as the algebraic degree of $\tr(aF_{k, j})$. We have shown that the image of each $F_{k, j}$ is not contained in any proper subspace of $\F_{2^n}$, hence $\tr(aF_{k, j})$ is never the zero function. Therefore $\degalg \tr(aF_{k, j}) = \degalg F_{k, j}$ for all $a \neq 0$ and since $\degalg F_{k, j}=k$ for $k<n$, this implies that $F_{k, j}$ is non-degenerate.
\end{proof}

\section{Equivalence of $k$th-order sum-free functions and Reed-Muller subcodes}
\label{sec:reed-muller}

In this section, we prove our main result, \cref{T: Equivalence}, which establishes an equivalence between large linear subcodes of $\RM(r,n)$ and $(n-r)$th-order sum-free functions. As a consequence, we obtain explicit constructions of large subcodes of $\RM(r,n)$ for any $2 \leq r \leq n-2$ which do not contain any minimum weight codewords of $\RM(r,n)$, see \cref{cor:construction}. We restate our main result.

\begin{manualtheorem}{T: Equivalence}
     Let $r,m, n$ be positive integers satisfying $2 \leq r \leq n-2$ and $m \leq n$. Any codimension $m$ (linear) subcode of the Reed-Muller code $\RM(r,n)$, which has minimum distance $3\cdot 2^{n-r-1}$, gives rise to a non-degenerate $(n-r)$th-order sum-free $(n, m)$-function. Conversely, any  non-degenerate $(n-r)$th-order sum-free $(n, m)$-function yields a codimension~$m$ (linear) subcode of $\RM(r,n)$ with minimum distance $3 \cdot 2^{n-r-1}$.
\end{manualtheorem}

We prove \cref{T: Equivalence} by a series of lemmas and propositions. First, we show that subcodes of $\RM(r,n)$, which do not contain any minimum weight codewords, give rise to $(n-r)$th-order sum-free functions.

\begin{lem}\label{L: Extract function}
    Let $r$ and $n$ be positive integers satisfying $2 \leq r \leq n-2$. Let $\mathcal{C} \subseteq \RM(r, n)$ be a linear subcode of codimension $m$, such that $\mathcal{C}$ does not contain any minimum weight codewords of $\RM(r,n)$. Then, there exists a non-degenerate $(n-r)$th-order sum-free $(n, m)$-function which can be constructed from $\mathcal{C}$. 
\end{lem}

\begin{proof}
    Let $H$ be a parity-check matrix of $\RM(r,n)$, which is a generator matrix of $\RM(n-r-1,n)$. If $\mathcal{C}$ is a subcode of $\RM(r,n)$ of codimension $m$, then there exist $m$ linearly independent vectors $v_1, \dots, v_m$ such that $\spann(v_1,\dots,v_m) \cap \RM(n-r-1,n) = \set{0}$ which we may append to $H$ to create $H_\mathcal{C}$, the parity-check matrix of $\mathcal{C}$. 
    Let $M$ be the $m \times 2^n$ matrix whose rows are $v_1, \dots, v_m$, such that
    \[
    H_\mathcal{C} = \begin{bmatrix} 
    H \\ 
    M
    \end{bmatrix}.
    \]
    Note that for $f\in \text{RM}(r, n) \setminus \mathcal{C}$, we have $Hf=0$ and $Mf \neq 0$. As the columns of $H$ and $M$ are indexed by the vectors of $\F_{2}^n$, let $M_v$ denote the column of~$M$ corresponding to $v \in \F_2^n$. Now define an $(n,m)$-function $F$ by $F(v) = M_v$. We claim that $F$ is $(n-r)$th-order sum-free. Recall from \cref{L: min-weights} that the minimum weight codewords of $\RM(r,n)$ are precisely the incidence functions of the $(n-r)$-flats of $\Fpnv$. Let $A \subseteq \F_2^n$ be an $(n-r)$-flat, and denote the corresponding codeword by $f_A \in \RM(r,n)\setminus\mathcal{C}$.  Consequently, 
    \[
    \sum _{x \in A} F(x)  = Mf_A \neq 0,
    \]
    and $F$ is $(n-r)$th-order sum-free.

    We now prove that $F$ is non-degenerate. Recall that the rows of~$M$ are linearly independent and disjoint from $\RM(n-r-1,n)$. Hence, no nontrivial linear combination of the rows of~$M$ can be contained in $\RM(n-r-1,n)$ or, equivalently, no component function of $F$ has algebraic degree less than $(n-r)$. This completes the proof.
\end{proof}

\begin{remark} 
\cref{T: m lower bound} implies that if such a linear subcode $\mathcal{C}$ exists, then its codimension $m$ must satisfy  $m \geq \max\{n-r+2, r+2\}$.
\end{remark}

\begin{prop}\label{P: Code bound}
    Let $r$ and $n$ be positive integers satisfying $2 \leq r \leq n-2$. Suppose $\mathcal{C}\subseteq RM(r, n)$ is a subcode of codimension $m \leq n$. Then $\mathcal{C}$ has minimum distance $d(\mathcal{C}) \leq 3\cdot 2^{n - r -1}$. 
\end{prop}

\begin{proof}
    Recall from \cref{L: min-weights} that the minimum weight codewords are the incidence vectors of $(n-r)$-flats and the second weight codewords are precisely the symmetric difference of two minimum weight codewords whose corresponding $(n-r)$-flats intersect in an $(n-r-2)$-flat.
    
    By assumption, the codimension of $\mathcal{C}$ in $\RM(r,n)$ is $m \leq n$. Suppose for a contradiction that $d(\mathcal{C}) > 3\cdot 2^{n-r-1}$. By Lemma \ref{L: Extract function}, there exists an $(n-r)$th-order sum-free $(n, m)$-function~$F$ which corresponds to $\mathcal{C}$. We partition the Grassmannian of all $(n-r)$-spaces in~$\F_2^n$ according to their witness under the function $F$.
    Assume that there are two $(n-r)$-spaces which intersect in an $(n-r-1)$-space or an $(n-r-2)$-space such that their witnesses are equal. Now, the sum of their corresponding codewords is a codeword in $\mathcal{C}$ of weight at most $3\cdot 2^{n-r-1}$, but $d(\mathcal{C}) > 3\cdot 2^{n-r-1}$, a contradiction.
    Hence, we have a valid coloring of the generalized Grassmann graph $J_2(n, n-r, n-r-2)$ with $\Gauss{m}{1}$ colors. First consider the case $r<n-2$. Then, such a coloring, together with \cref{P: general bound chromatic number} implies
\[
  \max\left\{ \Gauss{r+2}{2}, \Gauss{n-r+2}{2} \right\}\le \chi(J_2(n, n-r, n -r-2)) \le \Gauss{m}{1}.
\]
 We note that the left-hand side is always at least $\Gauss{ \lceil\frac{n}{2}\rceil + 2}{2}$. Therefore,
$$
\frac{(2^{\frac{n}{2} + 2}-1)(2^{\frac{n}{2} + 1} - 1)}{3} \leq \Gauss{ \lceil\frac{n}{2}\rceil + 2}{2} \leq \Gauss{m}{1} \leq \Gauss{n}{1} = 2^n - 1,
$$
but $\frac{(2^{\frac{n}{2} + 2}-1)(2^{\frac{n}{2} + 1} - 1)}{3}>2^n - 1$. Hence, we obtain a contradiction.\par

Now, consider the case $r=n-2$. Then we obtain a valid coloring of the generalized Grassmann graph $J_2(n, 2, 0)$ with $\Gauss{m}{1}$ colors. However, according to \cref{P: general bound chromatic number}, we have $\chi(J_2(n, 2, 0)) = \Gauss{n}{2}$. As $m\leq n$, this is also a contradiction.
\end{proof}

\noindent \textbf{The Reed-Muller Subcode}: 
We can construct a subcode~$\mathcal{C}_F$ of $\RM(r,n)$ from a non-degenerate $(n-r)$th-order sum-free $(n,m)$-function $F$ as follows: Denote by $H$ a parity-check matrix of $\RM(r,n)$ (which is precisely a generator matrix of $\RM(n-r-1,n)$) and define $\mathcal{C}_F$ as the code with parity-check matrix
\[
\begin{bmatrix}
    H\\
    F(x)
\end{bmatrix}_{x\in\Fpnv}.
\]
For convenience, denote the $m\times2^n$ matrix~$[F(x)]_{x\in\Fpnv}$ by $M_F$ in the following proofs.

\begin{lem}
    Let $\mathcal{C}_F$ be as above. Then
    $
    \dim(\mathcal{C}_F) = \dim(\RM(r, n)) - m.
    $
\end{lem}

\begin{proof}
     It is immediate that $\dim(\mathcal{C}_F) \geq \dim(RM(r, n)) - m$. Since $F$ is non-degenerate, note that the rows of $M_F$ correspond to the evaluation of a set of $m$ linearly independent Boolean functions of algebraic degree at least $n - r$. Since the dual code of $\RM(r,n)$ is RM$(n - r- 1, n)$, our observations above imply $\text{row}(M_F)\cap RM(n-r-1, n) = \{0 \}$. Therefore,
     \begin{equation*}
    \dim(\mathcal{C}_F)
    = \dim(\RM(r, n)) - \operatorname{rank}(M_F)
    = \dim(\RM(r, n)) - m.
    \qedhere
\end{equation*}
\end{proof}

\begin{prop}\label{P: New Code}
    Let $n, m, r$ be positive integers with $2 \leq r \leq n-2$. Let $F$ be an $(n-r)$th-order sum-free $(n, m)$-function and let $\mathcal{C}_F \subseteq \RM(r, n)$ be the corresponding Reed-Muller subcode. Then $\mathcal{C}_F$ has minimum distance $d(\mathcal{C}_F) \geq 3\cdot 2^{n-r - 1}$. 
\end{prop}

\begin{proof}
     To prove our claim, we must show that no minimum weight codewords of $\RM(r, n)$ are contained in $\mathcal{C}_F$.  If $f$ is a minimum weight codeword of $\RM(r, n)$, then, by \cref{L: min-weights}, the support of $f$ corresponds precisely to some $(n-r)$-flat $A$ of $\F_2^n$. As $F$ is an $(n-r)$th-order sum-free function, we have
$$
 M_F\cdot f = \sum_{x\in A}F(x) \neq 0.
$$
 Hence, $d(\mathcal{C}_F)$ is at least the second minimum weight of codewords in $\RM(r,n)$, which according to Theorem \ref{T: min-weights}, is exactly $3\cdot2^{n - r - 1}$ for $2 \leq r \leq n-2$.
\end{proof}

 Applying Propositions \ref{P: Code bound} and \ref{P: New Code} to the code $\mathcal{C}_F$ completes the proof of Theorem \ref{T: Equivalence}.    
 Given that $(n-r)$th-order sum-free $(n,n)$-functions exist for all positive integers $r, n$ satisfying $r < n$, we immediately obtain a corollary which yields the construction of a new family of subcodes of the Reed-Muller codes.

\begin{cor}
\label{cor:construction}
    Let $r, n$ be positive integers satisfying $2 \leq r \leq n-2$ and define $F(x) = x^{2^{n-r}-1}$. Then $\mathcal{C}_F \subseteq \text{RM}(r, n)$ has codimension $n$ and minimum distance $3\cdot 2^{n-r - 1}$.
\end{cor}

\begin{remark}
    Note that our construction of a subcode of $\RM(r,n)$ generalizes the construction of codes from an APN function $F$ with parity-check matrix
    \[
        H_F = \begin{bmatrix}
            1\\
            x\\
            F(x)
        \end{bmatrix}_{x\in\Fpnv},
    \]
    because the two upper blocks of $H_F$ form a parity-check matrix of $\RM(n-2,n)$. These codes are not only interesting on their own, as they have optimal properties if $F$ is APN, but they are heavily relied on in testing computationally whether two arbitrary $(n,m)$-functions are CCZ-equivalent: Two $(n,m)$-functions $F,G$ are CCZ-equivalent if and only if the codes generated by $H_F$ and $H_G$ are equivalent~\cite{edel2009}. We refer to \cite{pott_survey} for more background on these codes. It will be interesting to study the connection between
    the equivalence of $k$th-order sum-free functions $F,G$ 
    and the equivalence of their associated codes $\mathcal{C}_F,\mathcal{C}_G$. Clearly, degree-$(k-1)$ equivalent $k$th-order sum-free functions produce equivalent codes, and we confirmed computationally that the third-order sum-free functions $x^7$ and $x^{21}$ in dimension $5$ from \cref{prop:Carlet_kth-order} produce inequivalent codes. But it is an open question whether or not code equivalence in this setting is actually more general than degree-$(k-1)$ equivalence.
\end{remark}

\section{New colorings of the Grassmann graph}
\label{sec:coloring}

In this section, our goal is to prove the following theorem. The coloring that has the desired property will be defined in \cref{D: coloring function}.

\begin{thm} \label{T: k and k-1-order sum-free gives coloring}
    Let $k$ and $n$ be positive integers satisfying $2\leq k < n$. Assume that there exists an $(n,m)$-function $F$, such that $F$ is 
 of algebraic degree $k$
and $\set{k-1,k}$-order sum-free.
    Then, there exists a coloring of the Grassmann graph $J_2(n+1,k)$ with $\Gauss{m}{1}$ colors.
\end{thm}

As the chromatic number of the Grassmann graph is rather trivial for $k<2$, the condition on $k$ is not restrictive.
Recall that  for a flat $A\subseteq\Fpnv$ and an $(n,m)$-function $F$, we write $ \omega(A)=\sum_{x\in A} F(x)$. Our first lemma follows immediately from Proposition \ref{prop:degree<k}, which implies that any $(n, m)$-function of algebraic degree at most $k$ always sums to zero over any $(k+1)$-flat.

\begin{lem} \label{L: special condition} 
    Let $F$ be an $(n,m)$-function, let $k\geq 2$ and let $x, y \in \F_{2}^n$. Let $U$ be a $(k-1)$-space of $\F_2^n$. If $F$ has algebraic degree $k$, then $$\omega(U)+\omega(x+U)+\omega(y+U)=\omega(x+y+U).$$
    Furthermore, if $F$ is $(k-1)$th-order sum-free, then 
    $$
    \omega(U)+\omega(x+U)+\omega(y+U) \neq 0.
    $$
\end{lem}

We now set some notation. Identify the vector space $\F_2^{n+1}$ with $\F_2^n\times \F_2$. Let $\mathcal{H} = \F_2^n \times \{ 0\} \subseteq \F_2^{n+1}$ be a hyperplane. Now, we can understand a vector $\overline{v}$ of $\F_2^{n+1}$ as a pair $(v,v_1)$, where $v\in \F_2^n$ and $v_1\in \F_2$.

\begin{notation}
    Let $F$ be an $(n,m)$-function and let $U$ be a flat of $\F_2^n\times \F_2$. We write 
    \begin{align*}
        \Omega_F(U)=\sum\limits_{(v,v_1)\in U}^{}   F(v).
    \end{align*}
    If the function $F$ is clear from the context, we simply write $\Omega(U)$.
\end{notation}

\begin{remark}
    Note that this is subtly different from the witness notation introduced in \ref{Notation: Witness}.
\end{remark}

For any $k$-space $U$ of $\F_2^{n+1}$, we either have $U\subseteq \HH$, or $U\cap \HH$ is a $(k-1)$-space. 
If $U\subseteq \HH$, then $U$ is a $k$-space of $\HH$ and we can understand $\Omega(U)$ as $\omega(U)$, where $\Omega(U)$ is on $\F_2^{n+1}$ and $\omega(U)$ is on $\F_2^n$.

Now, we define the coloring function $c_F$.

\begin{defn} \label{D: coloring function}
     Let $F$ be an $(n,m)$-function that is $\{k-1,k\}$-order sum-free and of algebraic degree $k$ for $k\geq 2$. We define the function $c_F$ from the set of $k$-spaces of $\F_{2}^{n+1}$ to $\F_2^m$ as follows
 \begin{align*}
c_F(U) &=
\begin{cases}
\Omega(U) & \text{if } U\subseteq \HH, \\
\Omega(U\setminus \HH) & \text{if } U\not\subseteq \HH.
\end{cases}
\end{align*}
\end{defn}

\begin{remark}
Note that for $k = 2$, this definition coincides with the coloring function given in~\cite{heering2025lineparallelismspgn2preparatalike}.
\end{remark}

\begin{prop} \label{P: constructive Prove of Them 5.1} 
   Let $F$ be an $(n,m)$-function that is $\{k-1,k\}$-order sum-free and of algebraic degree $k$ for $2\leq k<n$. Then $c_F$ induces a valid coloring on $J_2(n+1,k)$.
\end{prop} 

\begin{proof}
Let $U$ be a $k$-space of $\F_2^{n+1}$. As $F$ is $\{k-1,k\}$-order sum-free, we always have  $c_F(U)\ne0$.
   It remains to be proven that $k$-spaces that intersect in dimension $k-1$ have different colors, i.e. $c_F$ assigns them different values.
There are four possible cases that have to be checked; we treat them all separately. In all cases, let $U_1,U_2$ be $k$-spaces of $\F_2^{n+1}$ with $\dim(U_1\cap U_2)=k-1$.

    \noindent \underline{Case 1:} Assume that $U_1,U_2\subseteq \HH$. Then 
    $$
    c_F(U_1) + c_F(U_2) = \Omega(U_1) + \Omega(U_2) =\omega(U_1) + \omega(U_2).
    $$
 \cref{L: kth-order coloring} implies that the quantity above is non-zero.

    \noindent \underline{Case 2:}  Assume that $U_1\subseteq \HH$  and  $U_2\not \subseteq \HH$. Then $U_1 \cap U_2 = W \subseteq \mathcal{H}$. 
    Then, there exist  $a,b\in \F_2^n$, such that 
    \begin{align*}
         U_1 &= W \cup (W + (a, 0)), \\
         U_2 &= W \cup (W+(b, 1)).
     \end{align*}
    By definition, we have 
    \begin{align*}
         c_F(U_1) + c_F(U_2)  &=\Omega(U_1) + \Omega(U_2\setminus \HH) \\
      &= \omega(W) + \Omega(W +(a, 0)) + \Omega(W +(b,1))\\
      &= \omega(W) + \omega(W +a) + \omega(W +b).
     \end{align*}
By Lemma \ref{L: special condition}, the quantity above is non-zero, and hence $c_F(U_1) \neq c_F(U_2)$.

     \noindent \underline{Case 3:} Assume that $U_1\cap U_2\subseteq \HH$ and $U_1,U_2\not\subseteq \HH$. Then we have
     $$
     c_F(U_1) +c_F(U_2) = \Omega(U_1\setminus \HH) + \Omega(U_2 \setminus \HH).
     $$
     Since $\dim(U_1 \cap U_2) = k-1$, then $(U_1 \setminus \HH) \cup (U_2\setminus \HH)$ is a $k$-flat, the projection on $\HH$ is also a $k$-flat and hence the sum above cannot be zero, as $F$ is $k$th-order sum-free. 

     \noindent \underline{Case 4:} Assume that $U_1,U_2\not\subseteq \HH$ and $U_1\cap U_2\not\subseteq \HH$. Let $W:= U_1 \cap U_2 \cap \HH$. Then, there exist  $a, b, c \in \F_2^n$ with distinct cosets modulo $W$, such that
     \begin{align*}
         U_1 &= W \cup (W + (a, 1)) \cup (W+ (b, 1)) \cup (W + (a+b, 0)), \\
         U_2 &= W \cup (W+(a, 1)) \cup (W + (c, 1)) \cup (W+ (a+c, 0)).
     \end{align*}
     Hence, 
     \begin{align*}
     c_F(U_1) + c_F(U_2)& = \Omega(W + (a, 1)) + \Omega(W + (b, 1)) + \Omega(W + (a, 1)) + \Omega(W + (c, 1)) \\
     &= \Omega(W + (b, 1)) + \Omega(W + (c, 1)).
     \end{align*}
     Since $(W + b)\cup (W + c)$ is a $(k-1)$-flat of $\HH$ and $F$ is $(k-1)$th-order sum-free, this quantity is non-zero.
 \end{proof}

The previous proposition directly implies \cref{T: k and k-1-order sum-free gives coloring}.

\section{On multiorder sum-free functions}
\label{sec: multiorder sum-free functions}

We showed in \cref{T: k and k-1-order sum-free gives coloring} that the existence of a $\set{k-1,k}$-order sum-free function of algebraic degree $k$ implies the existence of a coloring of $J_2(n+1,k)$.
For $k=2$, there are many known examples \cite{heering2025lineparallelismspgn2preparatalike} and in this case the coloring is optimal, see Section \ref{subsec:graphs}.
The inverse function $F_{inv}(x) = x^{2^n-2}$ in odd dimension $n$ is $ \set{1, 2,n-2,n-1}$-order sum-free, see \cite{carlet2025inverse}. 
So for $k=n-1$ this function satisfies the conditions of \cref{T: k and k-1-order sum-free gives coloring} and thus provides a coloring of $J_2(n+1,n-1)$. However, as $J_2(n+1,n-1)$ is dual to $J_2(n+1,2)$, this does not give new information regarding the chromatic number.

In the rest of this section, we summarize the results of our search for more such functions in small dimensions for $k\geq 3$. We start with a result that implies the existence of $\set{2,3}$-order sum-free functions for $n=5$. We show that for odd $n$, the family of $k$th-order sum-free functions from \cref{prop:Carlet_kth-order} always contains a multiorder sum-free function. We remark that the authors of \cite{blondeau2011} showed that $F_{m+1,1}$ is APN using the same argument.

\begin{prop}
\label{th:multiorder_carlet}
	Let $n$ be odd with $n=2m+1$. Let $i \in \set{1,\dots,2m}$ such that $\gcd(2i,n) = 1$. Then, the $(n,n)$-function $F_{m+1,2i}$ defined by
	\[
		F_{m+1,2i}(x) = x^\frac{2^{2i(m+1)}-1}{2^{2i}-1} = x^{1+2^{2i}+\dots+2^{2im}},
	\]
	where $2i$ is calculated modulo $n$, is the inverse of the Gold APN function $x^{2^i+1}$, and $F_{m+1,2i}$ is $\set{2,m+1}$-order sum-free.
\end{prop}

\begin{proof}
	According to \cref{prop:Carlet_kth-order}, $F_{m+1,2i}$ is $(m+1)$th-order sum-free. Since $n$ being odd and $\gcd(2i,n) = 1$ imply $\gcd(i,n) = 1$, the Gold function $x^{2^i+1}$ is APN. We prove that $F_{m+1,2i}$ is the inverse of $x^{2^i+1}$ by showing that their exponents multiply to $1$. Using $2i(m+1) = i(n+1)$, we multiply the exponents $2^i+1$ and $\frac{2^{i(n+1)}-1}{2^{2i}-1}$ and obtain $\frac{2^{i(n+1)}-1}{2^i-1}$. This expression equals $1 + 2^i +2^{2i} + \dots + 2^{(n-1)i} + 2^{ni}$ and, consequently, is congruent to $1 + \frac{2^{in}-1}{2^i-1} \pmod{2^n-1}.$ Since $2^{in}-1$ is divisible by $2^i-1$ and by $2^n-1$ and $\gcd(2^i-1,2^n-1)=2^{\gcd(i,n)}-1=1$, we have $\frac{2^{in}-1}{2^i-1} \equiv 0 \pmod{2^n-1}$. Taking the inverse is a CCZ-transformation, so $F_{m+1,2i}$ is APN or, equivalently, second-order sum-free.
\end{proof}

Therefore, the functions $x^7$ and $x^{21}$ are cubic and $\set{2,3}$-order sum-free in $\F_{2^5}$, and \cref{T: k and k-1-order sum-free gives coloring} yields the following.

\begin{cor}
    $ \chi(J_2(6,3))\leq 31$.
\end{cor}

Finally, we present the results of our computational search for more $\set{k-1,k}$-order sum-free functions in $\Fpnv$ for $n=5,6,7$. We checked if any of the known APN functions of algebraic degree at least $3$ are $k$th-order sum-free for $k \ge 3$. \par

For $n=5$, there are seven APN functions up to EA-equivalence, see \cite{brinkmannleander2008}. 
Except for the functions $x^7$ and $x^{21}$ and the inverse function $x^{30}$ mentioned above, none is third-order sum-free.\par

For $n=6$, all known APN functions are classified up to EA-equivalence \cite{calderini2020,langevin2012}. Recently, the authors of \cite{baudrin2025} conjectured that this list is complete. They list 716 EA-classes of APN functions, 521 of these are cubic and 182 of these have algebraic degree $4$. Our computations show that none of these functions is third- or fourth-order sum-free.\par

For $n=7$, the authors of \cite{baudrin2025} recently classified all APN functions that are CCZ-equivalent to one of the completely classified 488 CCZ-inequivalent quadratic APN functions~\cite{calderini2020} up to EA-equivalence. There are 60,302 cubic EA-classes and 21,619 EA-classes of algebraic degree $4$ in these 488 CCZ-classes. We confirmed computationally that none of these APN functions is third-order sum-free and that only the three functions $x^{15}, x^{77}$ and $x^{85}$ from \cref{th:multiorder_carlet}, that are equivalent to Gold APN functions, are fourth-order sum-free.

\section{Open problems}

All known $k$th-order sum-free functions are $(n,n)$-functions for $k\geq 1$. In Section \ref{sec:sum-free_functions} we show that for certain values of $m$ there are no $k$th-order sum-free functions which are $(n,m)$-functions. Nevertheless, the following is still open.

\begin{problem}
    Let $0<m, k < n$. Do there exist $(n, m)$-functions with $m < n$ that are $k$th-order sum-free?
\end{problem}

The existence, or non-existence of multiorder sum-free functions seems to be a nontrivial problem. Via \cref{T: k and k-1-order sum-free gives coloring}, a positive result has implications for the chromatic number of the Grassmann graph. The smallest open case is the following.

\begin{problem}
    Do there exist functions in $\F_{2^n}$ that are APN and third-order sum-free for $n>5$? If yes, is this true for infinitely many values of $n$?
\end{problem}

As we have seen in \cref{sec:reed-muller}, we may use $k$th-order sum-free functions to construct large subcodes of RM$(n-k, n)$. 
The Preparata codes of length $2^{n+1}$ can be constructed using crooked functions over $\F_{2^n}$ and all known crooked functions are quadratic APN permutations, which satisfy the conditions of  \cref{T: k and k-1-order sum-free gives coloring} with $k=2$. Hence, it is natural to ask the following.
\begin{problem}\label{prob: 2}
    Can $\{k-1, k\}$-order sum-free functions of degree $k$ be used to build new large non-linear codes?
\end{problem}
To motivate this a little more, note that when a Kerdock code of length $2^{n+1}$ exists, then $n+1$ is even, and hence the inverse function over $\F_{2^{n}}$ is $\{n-1, n-2\}$-order sum-free and of algebraic degree $(n-1)$. Similarly, there is one (seemingly sporadic) example of a non-linear binary code of length 64, minimum distance 12 and size $2^{37}$ constructed by Calderbank and McGuire \cite{1997CalderbankCode}. In this case, we also have a cubic $\{ 2, 3\}$-order sum-free function over $\F_{2^5}$.

At the end of \cref{sec:reed-muller}, we remarked that the equivalence of the codes associated with $k$th-order sum-free functions may be interesting to study. Therefore, we pose the following.
\begin{problem}
    What is the relation between code equivalence of the subcodes $\mathcal{C}_F$ and the corresponding sum-free $(n,m)$-functions $F$?
\end{problem}



\subsection*{AI Tool Disclosure}

ChatGPT Pro was used in the discovery and proof of Lemma \ref{L: kth-order-kth-degree}. 

\subsection*{Acknowledgement}

This paper was in part conceptualized during the conference Finite Geometries 2025 in Irsee, the authors would like to thank the organizers for a wonderful conference. 
The authors would also like to thank Alexander Pott and Leo Storme for many helpful discussions.
The first author would like to acknowledge the support of the Universiteit Gent, where part of this work was done while the first author was a visitor. 
The second author is funded by the Deutsche Forschungsgemeinschaft (DFG, German Research Foundation) – 541511634.

\bibliographystyle{plainurl}
\bibliography{bibliography-short-journals.bib}

@article{Meszka2013,
    AUTHOR = {Meszka, M.},
     TITLE = {The chromatic index of projective triple systems},
   JOURNAL = {J. Combin. Des.},
  FJOURNAL = {Journal of Combinatorial Designs},
    VOLUME = {21},
      YEAR = {2013},
    NUMBER = {11},
     PAGES = {531--540},
      ISSN = {1063-8539,1520-6610},
   MRCLASS = {05B07 (05C15)},
  MRNUMBER = {3103360},
MRREVIEWER = {Landang\ Yuan},
       DOI = {10.1002/jcd.21368}
}

@misc{dhaeseleer2026chromaticnumbergrassmanngraphs,
      title={Chromatic Number of {G}rassmann Graphs and {MRD} codes}, 
      author={J. D'haeseleer and F. Pavese and P. Santonastaso and V. Taranchuk},
      year={2026},
      eprint={2602.10777},
      archivePrefix={arXiv},
      primaryClass={math.CO}
}

@misc{subcodes_recursive,
      title={Subcodes of Second-Order {R}eed-{M}uller Codes via Recursive Subproducts}, 
      author={A P Vaideeswaran and Madireddi Sai Harish and Lakshmi Prasad Natarajan},
      year={2025},
      eprint={2501.10700},
      archivePrefix={arXiv},
      primaryClass={cs.IT}
}

@article{salomon,
author = {Salomon, A. and Amrani, O.},
year = {2005},
month = {12},
pages = {3918 - 3930},
title = {Augmented Product Codes and Lattices: {R}eed-{M}uller Codes and {B}arnes-{W}all Lattices},
volume = {51},
journal = {IEEE Trans. Inform. Theory},
doi = {10.1109/TIT.2005.856937}
}

@article {1997CalderbankCode,
    AUTHOR = {Calderbank, A. R. and McGuire, G. M.},
     TITLE = {Construction of a {$(64,2^{37},12)$} code via {G}alois rings},
   JOURNAL = {Des. Codes Cryptogr.},
  FJOURNAL = {Designs, Codes and Cryptography. An International Journal},
    VOLUME = {10},
      YEAR = {1997},
    NUMBER = {2},
     PAGES = {157--165},
      ISSN = {0925-1022,1573-7586},
   MRCLASS = {94B40 (11T71)},
  MRNUMBER = {1432295},
       DOI = {10.1023/A:1008240319733},
       URL = {https://doi.org/10.1023/A:1008240319733},
}

@book{MacWilliams1977TheTO,
    author    = {MacWilliams, F. J. and Sloane, N. J. A.},
  title     = {The Theory of Error-Correcting Codes},
  series    = {North-Holland Mathematical Library},
  volume    = {16},
  publisher = {North-Holland},
  address   = {Amsterdam},
  year      = {1977},
  pages     = {xii+762},
  isbn      = {978-0-444-85009-6},
  url = {https://www.sciencedirect.com/bookseries/north-holland-mathematical-library/vol/16/suppl/C}
}

@inproceedings{Jamali2022LowComplexityDO,
  title={Low-Complexity Decoding of a Class of {R}eed-{M}uller Subcodes for Low-Capacity Channels},
  author={M. V. Jamali and M. Fereydounian and H. Mahdavifar and H. Hassani},
  booktitle={ICC 2022 - IEEE International Conference on Communications},
  year={2022},
  pages={123-128},
  url={https://api.semanticscholar.org/CorpusID:246652576}
}

@ARTICLE{Construction_of_RM_subcodes,
  author={Van Wonterghem, J. and Boutros, J. J. and Moeneclaey, M.},
  journal={IEEE Commun. Lett.}, 
  title={On Constructions of {R}eed-{M}uller Subcodes}, 
  year={2018},
  volume={22},
  number={2},
  pages={220-223},
  doi={10.1109/LCOMM.2017.2772247}}

@misc{heering2025lineparallelismspgn2preparatalike,
      title={Line-parallelisms of {PG}$(n, 2)$ from {P}reparata-like codes}, 
      author={P. Heering and V. Taranchuk},
      year={2025},
      eprint={2508.19901},
      archivePrefix={arXiv},
      primaryClass={math.CO}
}

@misc{dhaeseleer2025chromaticnumbergrassmanngraphs,
      title={On the Chromatic Number of {G}rassmann Graphs}, 
      author={J. D'haeseleer and V. Taranchuk},
      year={2025},
      eprint={2505.22055},
      archivePrefix={arXiv},
      primaryClass={math.CO},
}

@book{algebraic_graph_theory,
  author     = {Godsil, C. and Royle, G.},
  title      = {{Algebraic Graph Theory}},
  volume     = {207},
  year       = {2001},
  pages      = {},
  isbn       = {978-0-387-95220-8},
  doi        = {10.1007/978-1-4613-0163-9}
}

@inproceedings{parallelisms_1973,
  author    = {Zaitsev, G. V. and Zinoviev, V. A. and Semakov, N. V.},
  title     = {{Interrelation of Preparata and Hamming codes and extension of Hamming codes to new double-error-correcting codes}},
  booktitle = {Proceedings of the 2nd International Symposium on Information Theory},
  year      = {1973},
  pages     = {257--263}
}

@inproceedings{eurocrypt-1993-2628,
  author     = {Nyberg, K.},
  title      = {Differentially Uniform Mappings for Cryptography},
  booktitle  = {Advances in Cryptology - EUROCRYPT '93, Workshop on the Theory and Application of of Cryptographic Techniques, Lofthus, Norway, May 23-27, 1993, Proceedings},
  series     = {Lecture Notes in Computer Science},
  volume     = {765},
  pages      = {55--64},
  year       = {1993},
  doi        = {10.1007/3-540-48285-7_6}
}

@article{pott_survey,
  author      = {Pott, A.},
  title       = {{Almost perfect and planar functions}},
  journal     = {Des. Codes Cryptogr.},
  year        = {2016},
  issue_date  = {January   2016},
  publisher   = {Kluwer Academic Publishers},
  address     = {USA},
  volume      = {78},
  number      = {1},
  issn        = {0925-1022},
  pages       = {141--195},
  numpages    = {55},
  doi         = {10.1007/s10623-015-0151-x},
  url         = {https://doi.org/10.1007/s10623-015-0151-x},
}

@ARTICLE{BVW,
  author   = {Baker, R. D. and van Lint, J. and Wilson, R.},
  title    = {{On the Preparata and Goethals codes}},
  journal  = {IEEE Trans. Inform. Theory},
  volume   = {29},
  number   = {3},
  year     = {1983},
  pages    = {342--345},
  doi      = {10.1109/TIT.1983.1056675}
}

@article {brinkmannleander2008,
	AUTHOR = {Brinkmann, M. and Leander, G.},
	TITLE = {On the classification of {APN} functions up to dimension five},
	JOURNAL = {Des. Codes Cryptogr.},
	FJOURNAL = {Designs, Codes and Cryptography. An International Journal},
	VOLUME = {49},
	YEAR = {2008},
	NUMBER = {1-3},
	PAGES = {273--288},
	ISSN = {0925-1022},
	MRCLASS = {94A60 (11T71 94C10)},
	MRNUMBER = {2438456},
	DOI = {10.1007/s10623-008-9194-6},
	URL = {https://doi.org/10.1007/s10623-008-9194-6},
}

@inproceedings{baudrin2025,
	TITLE = {Exploring the Set of {APN} Functions},
	AUTHOR = {Baudrin, J. and Galissant, P. and Perrin, L.},
	URL = {https://hal.science/hal-05376704},
	BOOKTITLE = {{BFA 2025 - 10th International Workshop on Boolean Functions and their Applications}},
	ADDRESS = {Larnaca, Cyprus},
	YEAR = {2025},
	MONTH = Sep,
	HAL_ID = {hal-05376704},
	HAL_VERSION = {v2},
}

@article {calderini2020,
    AUTHOR = {Calderini, M.},
     TITLE = {On the {EA}-classes of known {APN} functions in small
              dimensions},
   JOURNAL = {Cryptogr. Commun.},
  FJOURNAL = {Cryptography and Communications. Discrete Structures, Boolean
              Functions and Sequences},
    VOLUME = {12},
      YEAR = {2020},
    NUMBER = {5},
     PAGES = {821--840},
      ISSN = {1936-2447,1936-2455},
   MRCLASS = {94A60 (06E30 11T71 14G50)},
  MRNUMBER = {4141572},
MRREVIEWER = {Sugata\ Gangopadhyay},
       DOI = {10.1007/s12095-020-00427-1},
       URL = {https://doi.org/10.1007/s12095-020-00427-1},
}

@misc{langevin2012, 
	author={Langevin, P. and Saygi, E. and Saygi, Z.},
	url={https://langevin.univ-tln.fr/project/apn-6/apn-6.html}, 
	title={Classification of {APN} cubics in dimension 6 over {GF(2)}}, 
	year={2012}
}

@book {carlet2021book,
	AUTHOR = {Carlet, C.},
	TITLE = {Boolean functions for cryptography and coding theory},
	PUBLISHER = {Cambridge University Press, New York},
	YEAR = {2020},
	PAGES = {xiv+562},
	ISBN = {978-1-108-47380-4; [9781108606806]},
	MRCLASS = {94A60 (06E30 11T71 94D10)},
	MRNUMBER = {4625791},
	DOI = {10.1017/9781108606806},
}

@article {1970Berlekamp,
    AUTHOR = {Berlekamp, E. R. and Sloane, N. J. A.},
     TITLE = {Restrictions on weight distribution of {R}eed-{M}uller codes},
   JOURNAL = {Inform. and Control},
  FJOURNAL = {Information and Control},
    VOLUME = {14},
      YEAR = {1969},
     PAGES = {442--456},
      ISSN = {0019-9958,1878-2981},
   MRCLASS = {94.10},
  MRNUMBER = {243897},
doi = {https://doi.org/10.1016/S0019-9958(69)90150-8},
}

@article {kasami1970,
    AUTHOR = {Kasami, T. and Tokura, N.},
     TITLE = {On the weight structure of {R}eed-{M}uller codes},
   JOURNAL = {IEEE Trans. Inform. Theory},
  FJOURNAL = {Institute of Electrical and Electronics Engineers.
              Transactions on Information Theory},
    VOLUME = {IT-16},
      YEAR = {1970},
     PAGES = {752--759},
      ISSN = {0018-9448,1557-9654},
   MRCLASS = {94.15},
  MRNUMBER = {277300},
       DOI = {10.1109/tit.1970.1054545},
       URL = {https://doi.org/10.1109/tit.1970.1054545},
}

@article {carlet2025generalizations,
	AUTHOR = {Carlet, C.},
	TITLE = {Two generalizations of almost perfect nonlinearity},
	JOURNAL = {J. Cryptology},
	FJOURNAL = {Journal of Cryptology. The Journal of the International
	Association for Cryptologic Research},
	VOLUME = {38},
	YEAR = {2025},
	NUMBER = {2},
	PAGES = {Paper No. 20, 32},
	ISSN = {0933-2790},
	MRCLASS = {94A60},
	MRNUMBER = {4870858},
	DOI = {10.1007/s00145-025-09538-5}
}

@article {carlet2025inverse,
	AUTHOR = {Carlet, C.},
	TITLE = {On the vector subspaces of {$\mathbb {F}_{2^n}$} over which the
	multiplicative inverse function sums to zero},
	JOURNAL = {Des. Codes Cryptogr.},
	FJOURNAL = {Designs, Codes and Cryptography. An International Journal},
	VOLUME = {93},
	YEAR = {2025},
	NUMBER = {4},
	PAGES = {1237--1254},
	ISSN = {0925-1022},
	MRCLASS = {11T06 (12E05 12E10)},
	MRNUMBER = {4894915},
	DOI = {10.1007/s10623-024-01531-6}
}

@misc{carlet2025tDegree,
	author = {C. Carlet},
	title = {A notion on {S}-boxes for a partial resistance to some integral attacks},
	howpublished = {Cryptology {ePrint} Archive, Paper 2024/1693},
	year = {2025},
	url = {https://eprint.iacr.org/2024/1693}
}

@article {hou2006,
	AUTHOR = {Hou, X-d.},
	TITLE = {Affinity of permutations of {$\mathbb F_2^n$}},
	JOURNAL = {Discrete Appl. Math.},
	FJOURNAL = {Discrete Applied Mathematics. The Journal of Combinatorial
	Algorithms, Informatics and Computational Sciences},
	VOLUME = {154},
	YEAR = {2006},
	NUMBER = {2},
	PAGES = {313--325},
	ISSN = {0166-218X},
	MRCLASS = {06E30 (94B05)},
	MRNUMBER = {2194404},
	MRREVIEWER = {Marcel Wild},
	DOI = {10.1016/j.dam.2005.03.022},
	URL = {https://doi.org/10.1016/j.dam.2005.03.022},
}

@misc{kaspers2026WCC,
	title={Nonvanishing $k$-flats of {B}oolean and vectorial functions}, 
	author={C. Kaspers},
	year={2026},
	eprint={2603.28266},
	archivePrefix={arXiv},
	primaryClass={math.CO}
}

@article {blondeau2011,
	AUTHOR = {Blondeau, C. and Canteaut, A. and Charpin, P.},
	TITLE = {Differential properties of {$x\mapsto x^{2^t-1}$}},
	JOURNAL = {IEEE Trans. Inform. Theory},
	FJOURNAL = {Institute of Electrical and Electronics Engineers.
	Transactions on Information Theory},
	VOLUME = {57},
	YEAR = {2011},
	NUMBER = {12},
	PAGES = {8127--8137},
	ISSN = {0018-9448},
	MRCLASS = {11L05 (94A60)},
	MRNUMBER = {2895385},
	DOI = {10.1109/TIT.2011.2169129},
	URL = {https://doi.org/10.1109/TIT.2011.2169129},
}

@article {ebeling2024,
	AUTHOR = {Ebeling, A. and Hou, X-d. and Rydell, A. and
	Zhao, S.},
	TITLE = {On sum-free functions},
	JOURNAL = {Finite Fields Appl.},
	FJOURNAL = {Finite Fields and their Applications},
	VOLUME = {110},
	YEAR = {2026},
	PAGES = {Paper No. 102744, 32},
	ISSN = {1071-5797,1090-2465},
	MRCLASS = {11G25 (11T06 11T71 94D10)},
	MRNUMBER = {4975424},
	DOI = {10.1016/j.ffa.2025.102744},
	URL = {https://doi.org/10.1016/j.ffa.2025.102744},
}

@article {houzhao2025a,
	AUTHOR = {Hou, X-d. and Zhao, S.},
	TITLE = {Two absolutely irreducible polynomials over {$\Bbb{F}_2$} and
	their applications to a conjecture by {C}arlet},
	JOURNAL = {Finite Fields Appl.},
	FJOURNAL = {Finite Fields and their Applications},
	VOLUME = {109},
	YEAR = {2026},
	PAGES = {Paper No. 102713, 17},
	ISSN = {1071-5797,1090-2465},
	MRCLASS = {11G25 (11T06 11T71 94D10)},
	MRNUMBER = {4942688},
	MRREVIEWER = {N.\ L.\ Manev},
	DOI = {10.1016/j.ffa.2025.102713},
	URL = {https://doi.org/10.1016/j.ffa.2025.102713},
}

@article {houzhao2025b,
	AUTHOR = {Hou, X-d. and Zhao, S.},
	TITLE = {On a conjecture about the sum-freedom of the binary
	multiplicative inverse function},
	JOURNAL = {Des. Codes Cryptogr.},
	FJOURNAL = {Designs, Codes and Cryptography. An International Journal},
	VOLUME = {94},
	YEAR = {2026},
	NUMBER = {4},
	PAGES = {Paper No. 86},
	ISSN = {0925-1022,1573-7586},
	MRCLASS = {11G20 (11T06 11T71 94D10)},
	MRNUMBER = {5058034},
	DOI = {10.1007/s10623-026-01819-9},
	URL = {https://doi.org/10.1007/s10623-026-01819-9},
}

@article{carletcharpinzinoviev1998,
	title = {Codes, bent functions and permutations suitable for {DES-like} cryptosystems},
	volume = {15},
	issn = {0925-1022},
	doi = {10.1023/A:1008344232130},
	number = {2},
	urldate = {2012-11-24},
	journal={Des. Codes Cryptogr.},
	author = {Carlet, C. and Charpin, P. and Zinoviev, V.},
	year = {1998},
	pages = {125{\textendash}156},
}

@ARTICLE{muller1954,
  author={Muller, D. E.},
  journal={Trans. I.R.E. Prof. Group Electron. Comput.}, 
  title={{Application of Boolean algebra to switching circuit design and to error detection}}, 
  year={1954},
  volume={EC-3},
  number={3},
  pages={6-12},
  doi={10.1109/IREPGELC.1954.6499441}}

@article {reed1954,
    AUTHOR = {Reed, I. S.},
     TITLE = {A class of multiple-error-correcting codes and the decoding scheme},
   JOURNAL = {Trans. IRE Prof. Group Inf. Theory},
  FJOURNAL = {Trans. IRE},
      YEAR = {1954},
    NUMBER = {PGIT-, PGIT-4},
     PAGES = {38--49},
   MRCLASS = {94.0X},
  MRNUMBER = {89789},
MRREVIEWER = {R.\ W.\ Hamming},
    DOI = {10.1109/TIT.1954.1057465}
}

@article{KoetterKschischang2008,
  author  = {Koetter, R. and Kschischang, F. R.},
  title   = {Coding for Errors and Erasures in Random Network Coding},
  journal = {IEEE Trans. Inform. Theory},
  volume  = {54},
  number  = {8},
  pages   = {3579--3591},
  year    = {2008},
  doi     = {10.1109/TIT.2008.926449}
}

@article{EtzionSilberstein2009Ferrers,
  author  = {Etzion, T. and Silberstein, N.},
  title   = {Error-Correcting Codes in Projective Spaces via Rank-Metric Codes and {Ferrers} Diagrams},
  journal = {IEEE Trans. Inform. Theory},
  volume  = {55},
  number  = {7},
  pages   = {2909--2919},
  year    = {2009},
  doi     = {10.1109/TIT.2009.2021376}
}

@article{SilbersteinEtzion2011Lexicodes,
  author  = {Silberstein, N. and Etzion, T.},
  title   = {Large Constant Dimension Codes and Lexicodes},
  journal = {Adv. Math. Commun.},
  volume  = {5},
  number  = {2},
  pages   = {177--189},
  year    = {2011},
  doi     = {10.3934/amc.2011.5.177}
}

@article{Etzion2015PartialKParallelisms,
  author  = {Etzion, T.},
  title   = {Partial {$k$}-Parallelisms in Finite Projective Spaces},
  journal = {J. Combin. Des.},
  volume  = {23},
  number  = {3},
  pages   = {101--114},
  year    = {2015},
  doi     = {10.1002/jcd.21392}
}

@incollection {edel2009,
	AUTHOR = {Edel, Y. and Pott, A.},
	TITLE = {On the equivalence of nonlinear functions},
	BOOKTITLE = {Enhancing Cryptographic Primitives with Techniques from Error
	Correcting Codes},
	SERIES = {NATO Sci. Peace Secur. Ser. D Inf. Commun. Secur.},
	VOLUME = {23},
	PAGES = {87--103},
	PUBLISHER = {IOS},
	ADDRESS = {Amsterdam},
	YEAR = {2009},
	MRCLASS = {05B10 (11T71 94A60)},
	MRNUMBER = {2762230},
	MRREVIEWER = {Qi Wang},
    doi = {10.3233/978-1-60750-002-5-87}
}

@incollection {langevinleander2008,
	AUTHOR = {Langevin, P. and Leander, G.},
	TITLE = {Classification of {B}oolean quartic forms in eight variables},
	BOOKTITLE = {Boolean functions in cryptology and information security},
	SERIES = {NATO Sci. Peace Secur. Ser. D Inf. Commun. Secur.},
	VOLUME = {18},
	PAGES = {139--147},
	PUBLISHER = {IOS, Amsterdam},
	YEAR = {2008},
	ISBN = {978-1-58603-878-6},
	MRCLASS = {94C10},
	MRNUMBER = {2581571},
    doi = {10.3233/978-1-58603-878-6-139}
}

@article {hou1996,
	AUTHOR = {Hou, X-d.},
	TITLE = {{${\textrm GL}(m,2)$} acting on {$R(r,m)/R(r-1,m)$}},
	JOURNAL = {Discrete Math.},
	FJOURNAL = {Discrete Mathematics},
	VOLUME = {149},
	YEAR = {1996},
	NUMBER = {1-3},
	PAGES = {99--122},
	ISSN = {0012-365X,1872-681X},
	MRCLASS = {94B05 (20B25)},
	MRNUMBER = {1375102},
	MRREVIEWER = {Jonathan\ I.\ Hall},
	DOI = {10.1016/0012-365X(94)00342-G},
	URL = {https://doi.org/10.1016/0012-365X(94)00342-G},
}

@article {abbe2021,
    AUTHOR = {Abbe, E. and Shpilka, A. and Ye, M.},
     TITLE = {Reed-{M}uller codes: theory and algorithms},
   JOURNAL = {IEEE Trans. Inform. Theory},
  FJOURNAL = {Institute of Electrical and Electronics Engineers.
              Transactions on Information Theory},
    VOLUME = {67},
      YEAR = {2021},
    NUMBER = {6, part 1},
     PAGES = {3251--3277},
      ISSN = {0018-9448,1557-9654},
   MRCLASS = {94A29 (94A24)},
  MRNUMBER = {4289318},
       DOI = {10.1109/TIT.2020.3004749},
       URL = {https://doi.org/10.1109/TIT.2020.3004749},
}

\end{document}